%% file: main.tex
\documentclass[journal=jnw]{CUP-JNL-DTM}%

\defbibenvironment{bibliography}
	{\list{\mkbibbrackets{\arabic{enumiv}}}%
		 {\usecounter{enumiv}%
			\settowidth{\labelwidth}{\mkbibbrackets{99}}%
			\setlength{\leftmargin}{\labelwidth}%
			\setlength{\labelsep}{\biblabelsep}%
			\addtolength{\leftmargin}{\labelsep}%
			\setlength{\itemsep}{\bibitemsep}%
			\setlength{\parsep}{\bibparsep}%
			}}
	{\endlist}
	{\item}

\usepackage[T1]{fontenc}
\usepackage{newtxtext}
\usepackage{newtxmath}
\usepackage{textcomp}

\usepackage{amsmath,amsfonts}

\usepackage{amsthm}

\usepackage{graphicx}
\usepackage{appendix}
\usepackage{lastpage}
\usepackage{changepage}

\usepackage{xcolor}
\usepackage[colorlinks,allcolors=blue]{hyperref}

\definecolor{apdx}{HTML}{098bb3}
\definecolor{ssharan_color}{HTML}{ad8242}
\definecolor{abbreviation_color}{HTML}{53735b}
\definecolor{revise_color}{HTML}{3455eb}

\newcommand{\ie}{\emph{i.e.}}

\newcommand{\hrunup}{h_{\rm run-up}}
\newcommand{\urunup}{u_{\rm run-up}}
\newcommand{\thrunup}{\tilde{h}_{\rm run-up}}

\newcommand{\cF}{{\cal F}}

\newtheorem{theorem}{Theorem}[section]

\theoremstyle{definition}

\newtheorem{proposition}[theorem]{Proposition}
\numberwithin{equation}{section}

\jname{Data/Math}
\articletype{RESEARCH ARTICLE}
\jyear{2026}

\usepackage{mathrsfs}

\begin{document}

\begin{Frontmatter}

\title[Article Title]{Dry dam-break over parabolic bathymetry}

\author[1]{Shashwat Sharan}
\author[1]{Patrick Sprenger}
\author[1]{Boaz Ilan}
\author[2]{Mark A. Hoefer}

\address[1]{\orgdiv{Department of Applied Mathematics},
\orgname{University of California Merced},
\orgaddress{\city{Merced}, \postcode{95343}, \state{CA}, \country{USA}}}

\address[2]{\orgdiv{Department of Applied Mathematics},
\orgname{University of Colorado Boulder},
\orgaddress{\city{Boulder}, \postcode{80309}, \state{CO}, \country{USA}}.
\par\noindent\textbf{Corresponding author:} Shashwat Sharan; Email:
\href{mailto:ssharan2@ucmerced.edu}{ssharan2@ucmerced.edu}}

\authormark{S. Sharan et al.}



\abstract{
Motivated by dry dambreak flows, a perturbative framework for solving the 1D shallow water equations (SWE) over parabolic bathymetry, together with the inviscid Burgers' equation, is developed. 
Expressed in Riemann variables, the solutions are expanded as analytic power series in the bathymetric curvature $\omega^2$ and summed in closed-form in terms of trigonometric functions.
The solutions, which are termed generalised rarefaction waves (GRWs), exhibit finite-time singularities and are shown to describe the asymptotic behaviour of the fluid near dry (vacuum) points.
The GRWs bear direct relevance to wave run-up at a shore.
A repulsive parabolic hill bathymetry is also considered, yielding solutions in terms of hyperbolic functions.
The connection to the nonlinear Schr\"odinger/Gross-Pitaevskii equation that models Bose-Einstein condensates confined in a harmonic potential is discussed.The analytical results agree with direct numerical simulations of the Burgers' equation and of the SWE in different scenarios.
}

\end{Frontmatter}

\setcounter{tocdepth}{2}
\localtableofcontents

\input{sections/1_introduction}

\input{sections/2_main_results}
\input{sections/3_burgers}
\input{sections/4_swe}
\input{sections/5_applications}
\input{sections/6_conclusion}

\addtocontents{toc}{\protect\setcounter{tocdepth}{1}}
\input{sections/appendix}
\input{backmatter}

\end{document}

%% file: sections/1_introduction.tex
\section{Introduction}

Nonlinear waves arising from dambreak flows have been studied mathematically since the late 19$^{\rm th}$ century~\cite{ritter1892fortpflanzung}, primarily in the context of water waves~\cite{stoker2019water}. More recently, analogous dynamics have been observed in nonlinear optics~\cite{xu2017dispersive,dieli2024observation} and Bose-Einstein condensates~\cite{sharan2025breaking}, where they are modelled by nonlinear Schr\"odinger (NLS) / Gross--Pitaevskii (GP)-type equations~\cite{el2016dam}.

The standard model for long-wavelength, nondispersive flows in a shallow fluid layer is the 1D shallow water equations (SWE) with variable bathymetry $b(x)$,
\begin{subequations}\label{eq:SWE-bathymetry}
\begin{align}
\frac{\partial h}{\partial t} + \frac{\partial (hu)}{\partial x} &= 0, \\[4pt]
\label{eq:SWE-bathymetry_u}
\frac{\partial u}{\partial t} + u\, \frac{\partial u}{\partial x}  + \frac{\partial h}{\partial x} &= -b'(x)~,
\end{align}
\end{subequations}
where $h(x,t)$ denotes the mean fluid depth, $u(x,t)$ is the depth-averaged horizontal velocity, and $-b'(x)$ accounts for varying seabed topography.
In this study we consider the parabolic profile
\begin{align}
\label{eq:b(x)}
    b(x)=\frac{1}{2}\omega^{2}x^{2}~,
\end{align}
 where $\omega^2$ parameterises the bathymetric curvature. 
 
For flat bathymetry, $b(x)={\rm const.}$, equations~\eqref{eq:SWE-bathymetry} form a system of conservation laws whose Riemann invariants facilitate the construction of simple wave solutions, in which all but one Riemann invariant is constant. Classical rarefaction waves arising from Riemann initial data are of this type and depend on space and time only through the ratio $x/t$, \ie, they are self-similar.
Recently, Camassa \emph{et al.}~\cite{camassa2019singularity,camassa2020vacuum} studied the 1D SWE with flat bathymetry for continuous initial data containing a dry point (vacuum state) together with a discontinuity in the slope.
They showed that, under the subsequent evolution, the depth gradient at the dry point can diverge in finite time, producing a gradient catastrophe.

For linear bathymetry, $b(x)=\alpha x$, the spatially constant forcing $-\alpha$ can be removed by passing to an accelerated frame~\cite{chirkunov2014exact,dorodnitsyn2021discrete}.
In this accelerated frame, the equations reduce to the flat-bathymetry SWE. The classical flat-bathymetry simple-wave construction of self-similar rarefaction waves can then be applied in the moving frame.
Exact solutions of the forced SWE are therefore known primarily for linear bathymetry, where the simple-wave theory remains available after this change of variables. Carrier and Greenspan applied a hodograph--Legendre transformation to describe wave run-up over a linearly sloping beach~\cite{carrier1958water}, following Stoker~\cite{stoker1948formation}. This approach has been extended in several studies~\cite{ovsyannikov1979two,pelinovsky1992exact,didenkulova2011nonlinear,ezersky2013physical,kanouglu2006initial}. However, when the bathymetry profile is nonlinear, no analogous reduction is available and both the simple-wave construction and the self-similar $x/t$ scaling of the flat problem are lost.

For nonlinear bathymetry, very few exact solutions are known. 
Thacker was the first to obtain an exact solution of the 2D SWE in an elliptic paraboloid basin~\cite{thacker1981some}. When reduced to the 1D setting, Thacker's solution is a very special one describing a sloshing fluid with a planar surface profile.
Aksenov and Druzhkov recently derived exact conservation laws for the 1D SWE with~\eqref{eq:b(x)} via Lie symmetries and Lagrange variables~\cite{aksenov2016conservation,aksenov2020conservation}.
For a compendium of analytic solutions of the shallow water equations, including the classical flat-bottom dam-break solutions and sloshing solutions over a parabolic bottom, we refer the reader to~\cite{delestre2013swashes}.

More recently, Camassa \emph{et al.}~\cite{camassa2022evolution} studied the SWE with variable bathymetry through a wave-front expansion, representing the solution as a power series in a coordinate centred on the moving front, which reduces the near-front dynamics to a hierarchy of ordinary differential equations (ODEs) for the time-dependent coefficients. For a parabolic bathymetry, they found that the system admits self-similar solutions of the second kind, in which the depth is quadratic in space and the velocity linear in space, whose coefficients evolve in time according to the equations. Studying the ODE system governing these coefficients, they characterised its dynamics, including periodic sloshing and finite-time blow-up of the curvature, and obtained a closed-form expression for the blow-up time in terms of elliptic functions, the solution itself being expressed implicitly through inverse elliptic functions.

To our knowledge no closed-form solution is known for the Riemann dam-break problem in this setting. We narrow this gap in the present work by deriving various special closed-form explicit solutions,  which we use to describe a Riemann dam-break problem and other fluid configurations over parabolic bathymetry.

%% file: sections/2_main_results.tex
\section[Main results]{Main results}\label{sec:main_results}

We consider~\eqref{eq:SWE-bathymetry} with the parabolic bathymetry~\eqref{eq:b(x)}, \ie, the forced SWE
\begin{subequations}\label{eq:swe-parabolic}
\begin{align}
\frac{\partial h}{\partial t} + \frac{\partial (hu)}{\partial x} &= 0, \\[4pt]
\frac{\partial u}{\partial t} + u\, \frac{\partial u}{\partial x}  + \frac{\partial h}{\partial x} &= -\omega^2 x~,
\end{align}
\end{subequations}
as well as the related forced Burgers' equation
\begin{equation}\label{eq:burgers-parabolic}
\frac{\partial u}{\partial t} + u \, \frac{\partial u}{\partial x} = -\omega^2 x~.
\end{equation}
We obtain closed-form dam-break solutions of Eq.~\eqref{eq:burgers-parabolic} and the forced SWE~\eqref{eq:swe-parabolic}, expressed in terms of elementary trigonometric functions. These solutions are constructed through a perturbative expansion in the bathymetric curvature $\omega^2$, which yields the solutions of the forced equations from their unforced counterparts. Although the parabolic forcing breaks the self-similar scaling that underlies the classical rarefaction wave theory, the perturbative corrections remain functions of the self-similar variable $x/t$ and acquire an explicit dependence on $t$. 

The framework is first illustrated on the forced Burgers' equation and then extended to the forced SWE. We term the resulting solutions \emph{generalised rarefaction waves} (GRWs), which reduce to their classical unforced counterparts in the flat-bathymetry limit $\omega\to 0$. These solutions fall within the class studied by Camassa \emph{et al.}~\cite{camassa2022evolution}, for which closed-form expressions were not previously available. 

The GRW solutions contain terms that become singular in finite time, common to both equations; we show that this singularity is never reached for physical fluid configurations with finite mass. In addition, we extend the construction to a repulsive parabolic hill bathymetry, where the fluid is expelled rather than trapped.
As noted in previous studies, cf.~\cite{liu1980vacuum,camassa2020vacuum}, the velocity field is not uniquely determined in the vacuum. Among these, the perturbative framework selects the rarefaction-vacuum branch whose characteristics fan out from the initial discontinuity, and remains continuous across the vacuum interface. This choice requires no velocity data to be prescribed in the vacuum region at the initial time.
For different types of dam-break initial data, we show that an exact SWE solution describes the behaviour of the fluid asymptotically near the dry edge. These results are corroborated by direct numerical simulations.

Below we present the GRW solutions for~\eqref{eq:burgers-parabolic}  and~\eqref{eq:swe-parabolic}. In both cases the solution describes the evolution of a rarefaction wave generated by Riemann initial data with a jump discontinuity at the origin. We also state a general translation law that extends these centred solutions to off-centred initial data.

\subsection[GRW solution for forced Burgers' equation]{GRW solution of the forced Burgers' equation}\label{sec:main_burgers}

We consider the forced Burgers' equation~\eqref{eq:burgers-parabolic} with a centred Riemann data with a jump at the origin,
\begin{align}
\label{eq:burgers-ic}
u(x,0) &=
\begin{cases}
u_{-}, & x < 0,\\
u_{+}, & x > 0,
\end{cases}
\qquad u_{-} < u_{+}.
\end{align}
We show that the solution of this problem evolves as
\begin{align}
\label{eq:burgers-sol-main}
u(x,t;\omega) &=
\begin{cases}
-\omega x \tan(\omega t) + u_{-}\sec(\omega t),
    & x < X_{-}(t;\omega), \\[2mm]
 \phantom{+}\omega x \cot(\omega t) 
    & X_{-}(t;\omega) \le x \le X_{+}(t;\omega), \\[2mm]
    -\omega x \tan(\omega t) + u_{+}\sec(\omega t),
    & x > X_{+}(t;\omega),
\end{cases}
\end{align}
where the edges separating the three regions are
\begin{align}\label{eq:burgers-edges}
X_{\pm}(t;\omega) &=
\frac{u_{\pm}}{\omega} \sin(\omega t).
\end{align}
The solution~\eqref{eq:burgers-sol-main} contains trigonometric terms that become singular at a finite time and defines a classical solution to the Burgers' equation only up to the first singularity time,
\begin{equation}\label{eq:singularity_time}
    t_* = \frac{\pi}{2\omega}.
\end{equation}
We also consider the more general case of off-centred initial data, with the initial jump in \(u\) located at \(x=x_0\).
Proposition~\ref{prop:translation}  provides a general translation theorem for parabolic bathymetry, allowing the off-centred solution to be obtained through a transformation of the centred solution.

\begin{figure}[ht!]
    \centering
    \includegraphics[width=0.9\linewidth]{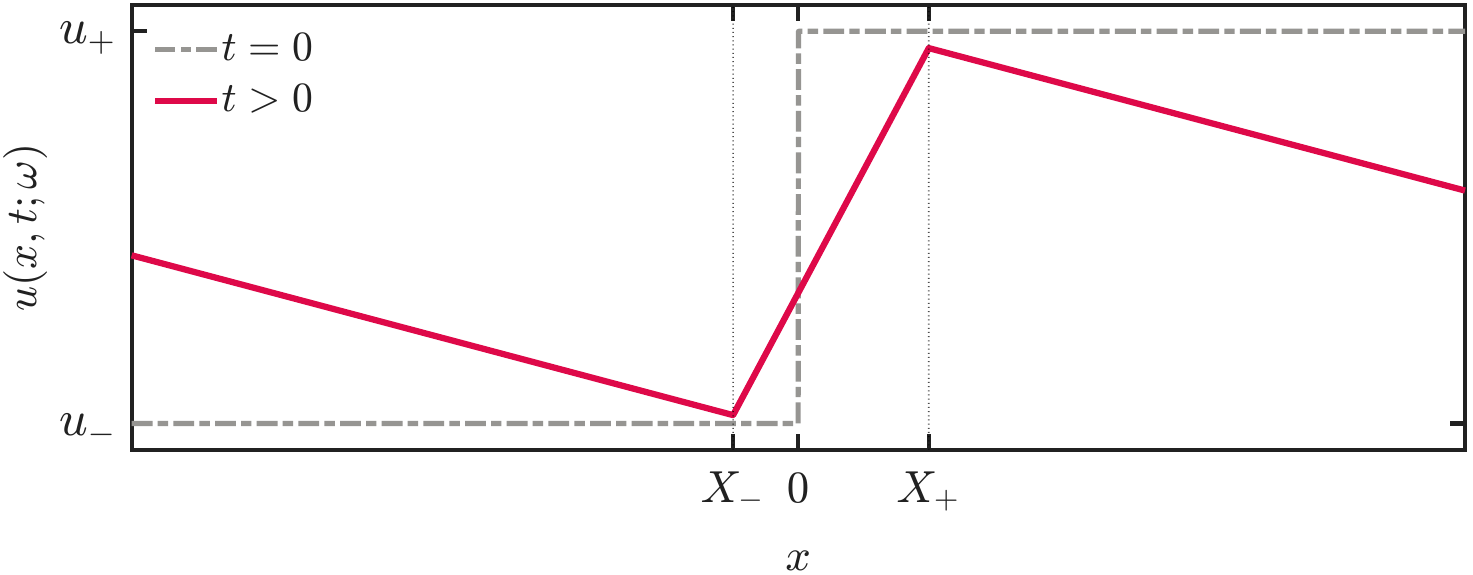}
    \caption{Centred GRW
        solution~\eqref{eq:burgers-sol-main} of the Burgers' equation. As time evolves, the left and the right branches become steeper; eventually becoming vertical as \(t\) approaches the singular time  \(t_* = \tfrac{\pi}{2\omega}\)}
    \label{fig:burgers-parabolic-solution}
\end{figure}

\subsection[GRW solution for shallow water equations]{GRW solution for SWE with parabolic bathymetry}\label{sec:main_swe}

Next, we consider the SWE with parabolic bathymetry~\eqref{eq:swe-parabolic} with centred Riemann initial data, in which the fluid depth has a jump discontinuity at the origin,
\begin{subequations}\label{eq:swe-dambreak-ic}
\begin{align}
h(x,0) &=
\begin{cases}
h_0, & x<0,\\
0, & x>0,
\end{cases}
\label{eq:swe-dambreak-ic-h}\\[4mm]
u(x,0) &= 0
\quad \text{for} \ x < 0.
\label{eq:swe-dambreak-ic-u}
\end{align}
\end{subequations}
The initial data consists of a fluid layer of constant depth \(h_0>0\) occupying the left half-line and vacuum on the right. No velocity data is prescribed in the vacuum region. This configuration corresponds to a \emph{dry} dambreak problem.

We show that the solution of this problem has the following structure:
\begin{subequations}\label{eq:SWE-dambreak-solution}
\begin{align}
h(x,t;\omega) &=
\begin{cases}
h_0 \sec(\omega t), & x < X_{\rm wet}(t;\omega), \\[2mm]
H_{\rm fan}(x,t;\omega), & X_{\rm wet}(t;\omega) < x < X_{\rm dry}(t;\omega), \\[2mm]
0, & x > X_{\rm dry}(t;\omega),
\end{cases}
\label{eq:SWE-dambreak-solution-h}
\\[4mm]
u(x,t;\omega) &=
\begin{cases}
-\omega x \tan(\omega t), & x < X_{\rm wet}(t;\omega), \\[2mm]
 U_{\rm fan}(x,t;\omega), & X_{\rm wet}(t;\omega) < x < X_{\rm dry}(t;\omega), \\[2mm]
\phantom{+} \omega x \cot(\omega t), & x > X_{\rm dry}(t;\omega).
\end{cases}
\label{eq:SWE-dambreak-solution-u}
\end{align}
\end{subequations}
Here \(H_{\rm fan}\) and \(U_{\rm fan}\) denote the nontrivial depth and velocity profiles inside the rarefaction fan, which match continuously with the exterior states.  The details of this solution are described below.

\paragraph{\textbf{Edges of the rarefaction fan.}}
The \emph{wet edge} $X_{\rm wet}(t;\omega)$ is the interface between the rarefaction fan and the spatially uniform depth region on the left. The \emph{dry edge} $X_{\rm dry}(t;\omega)$ is the leading front advancing into  the vacuum state on the right. Both satisfy $X_{\rm wet}(0;\omega)=X_{\rm dry}(0;\omega)=0$. These interfaces evolve according to
\begin{subequations}\label{eq:swe-edges}
    \begin{align}
    \label{eq:dry-edge}
X_{\rm dry}(t;\omega) &= \frac{2c_0}{\omega}\sin(\omega t), \\[2mm]
X_{\rm wet}(t;\omega) &= -\frac{c_0}{\omega} \cos(\omega t) \, \mathcal{F}(\omega t),
\end{align}
\end{subequations}
where
\begin{equation}
    \mathcal{F}(z) \doteq \int_{0}^{z} \sec^{3/2} (\theta) \, {\rm d}\theta.
\end{equation}
In the flat-bathymetry limit $\omega \to 0$, these reduce to the flat-bathymetry edges $X_{\rm dry} \to 2c_0\, t$ and $X_{\rm wet} \to -c_0\, t$. For $\omega > 0$, the dry edge follows a sinusoidal trajectory, decelerating due to the restoring effect of the parabolic well. The wet edge initially propagates leftward but reverses and returns to the origin as $\omega t \to (\pi/2)^{-}$. The derivation and a closed-form expression for the wet edge in terms of an incomplete elliptic integral are given in \S\ref{sec:swe_edges}. The edges are plotted in Figure~\ref{fig:swe-edges}.

\paragraph{\textbf{Asymptotic run-up solution near the dry edge.}}
Near the dry edge, the rarefaction fan has an asymptotic behaviour, which
describe the fluid run-up\footnote{Wave run-up is the maximum vertical elevation a fluid reaches as waves crash up a beach. In this work, we use this term more leniently to describe a fluid's behaviour near a dry edge.} over the parabolic bathymetry. As \(x\to X_{\rm dry}(t)^-\), the solution behaves as
\begin{equation}
\label{eq:main-run_up}
\begin{aligned}
H_{\rm fan}(x,t;\omega)
&\sim
\hrunup(x,t;\omega) \, \doteq \,
\frac{\omega^2}{4}\csc^2\Bigl(\frac{3\omega t}{2}\Bigr)\bigl(X_{\rm dry}(t;\omega)-x\bigr)^2,\\[2mm]
U_{\rm fan}(x,t;\omega)
&\sim
\urunup(x,t;\omega) \, \doteq \,
\frac{2c_0 + \omega x \cos(\frac{3\omega t}{2})\csc(\frac{\omega t}{2})}
{1+2\cos(\omega t)},
\end{aligned}
\end{equation}
where \(c_0 \doteq \sqrt{h_0}\).

\begin{figure}[h!]
    \centering
    \includegraphics[width=0.9\linewidth]{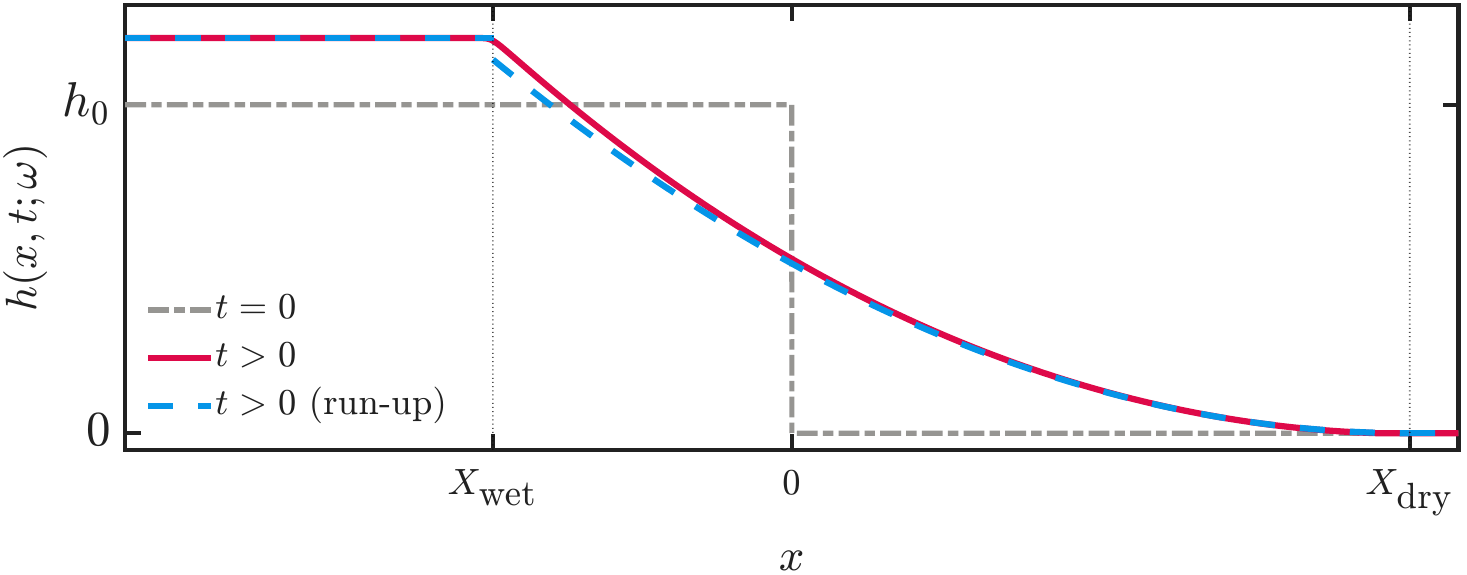}
    \caption{Depth profile \(h(x,t;\omega)\) of the centred dry dambreak solution~\eqref{eq:SWE-dambreak-solution} of the SWE over a parabolic bathymetry. The grey dash-dotted line shows the initial data~\eqref{eq:swe-dambreak-ic} at \(t=0\). The solid red line is the numerical solution at \(t>0\), and the blue dashed line 
    is the run-up solution \(\hrunup\)~\eqref{eq:main-run_up}}
    \label{fig:swe-parabolic-solution}
\end{figure}

\begin{figure}[h!]
    \centering
    \includegraphics[width=0.96\linewidth]{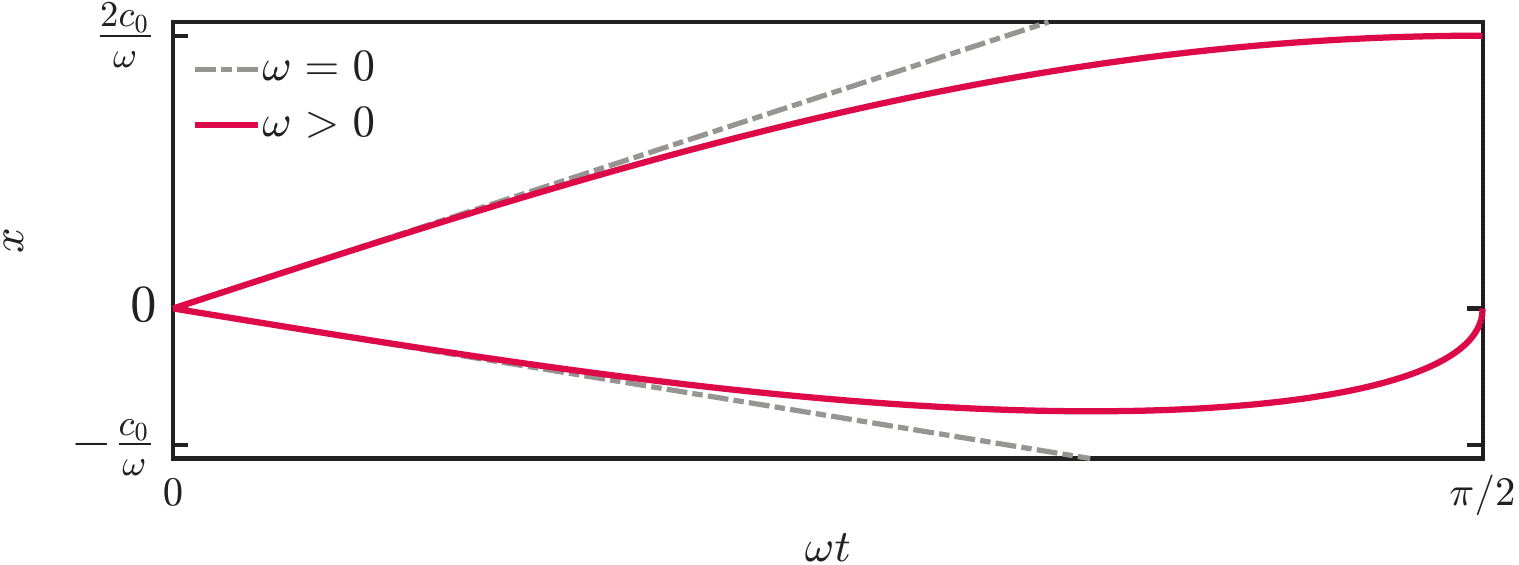}
    \caption{$X_{\rm dry}(t;\omega)$ and $X_{\rm wet}(t;\omega)$ from~\eqref{eq:swe-edges} compared with the flat-bathymetry edges ($\omega=0$)}
    \label{fig:swe-edges}
\end{figure}

Remarkably, the pair \((\hrunup,\urunup)\)~\eqref{eq:main-run_up} is an \emph{exact} solution of the forced SWE~\eqref{eq:swe-parabolic} in terms of elementary functions.
However, as we show in Section~\S\ref{sec:applications} for different initial value problems (IVPs), it arises as a \emph{local} description asymptotically near the dry edge. The asymptotic nature of this solution is discussed in Appendix~\ref{apdx_sec:wavefront_higher_order}.

We note that, in the flat-bathymetry limit \(\omega\to0\), the run-up solution~\eqref{eq:main-run_up} extends globally across the entire rarefaction fan, reducing to the flat-bathymetry dambreak solution~\eqref{eq:stoker_solution} discussed in \S\ref{sec:stoker_solution}.
For flat and linear bathymetries, the depth profile inside the rarefaction fan remains spatially quadratic for all times. Parabolic bathymetry is fundamentally different. In this case, the quadratic profile persists only as a local description near the vacuum edge (see Appendix~\ref{apdx:leading_runup}). Away from the edge, the full fan involves higher-order spatial corrections, and these corrections are not analytic in the bathymetric parameter $\omega$; in particular, they contain fractional powers of $\omega$ (see Appendix~\ref{apdx:higher_order}).

\paragraph{\textbf{Finite-time singularity.}}
The one-sided Riemann data~\eqref{eq:swe-dambreak-ic} has infinite mass and produces a finite-time singularity at the same time $t_*$~\eqref{eq:singularity_time} as in the Burgers case,
\begin{equation}
\label{eq:singularity_time_swer}
    \omega t_* = \frac{\pi}{2},
\end{equation}
at which the fluid depth in the $x < X_{\rm wet}$ region becomes infinite and the solution cannot be extended beyond this time. In Section~\S\ref{sec:applications} we consider fluid configurations with compact support and thus finite mass. For these physical configurations we show the finite-time singularity is never reached.

\paragraph{\textbf{Generalised translation law in parabolic bathymetry}}

The centred solutions above generate a corresponding family of off-centred solutions through the following translation law.
The general \emph{off-centred} case in which the jump in \(h\) and \(u\) are at \(x = x_0\) can be obtained from Proposition~\ref{prop:translation} whose proof is provided in Appendix~\ref{apdx_sec:proof_of_translation_proposition}.

\begin{adjustwidth}{1.5em}{0pt}
\begin{proposition}[\textbf{Off-centred solution}]
\label{prop:translation}
\itshape
Let \((h, u)\) be a solution of the forced SWE~\eqref{eq:swe-parabolic} with initial data  \(h_0(x) = h(x,0)\) and \(u_{0}(x)=u(x,0)\). Then the solution corresponding to the translated initial data \(h_{0}(x-x_0)\) and \(u_{0}(x-x_0)\) is given by
\begin{subequations}\label{eq:swe-shifted-transformation}
\begin{align}
h(x,t;\omega, x_0) &= h\bigl(x - x_0 \cos(\omega t),t\bigr)
\label{eq:swe-shifted-transformation-h}\\
u(x,t;\omega,x_0)  &= u\bigl(x-x_0\cos(\omega t),\,t\bigr)
-\omega x_0\sin(\omega t).
\label{eq:swe-shifted-transformation-u}
\end{align}
\end{subequations}
In the limit \(\omega \to 0\), the time-dependent shift reduces to a standard spatial translation, recovering the translational invariance of the unforced system.

\medskip 

\noindent \textbf{Remark:} A corresponding result for the forced Burgers' equation~\eqref{eq:burgers-parabolic} follows as a special case by setting \(h=0\), in which case~\eqref{eq:swe-shifted-transformation-h} drops out.
\end{proposition}
\end{adjustwidth}

\medskip

\medskip
\noindent The remainder of the paper is organised as follows.
Section \S\ref{sec:burgers} reviews the classical rarefaction wave for the unforced Burgers' equation and then derives the GRW solution presented in \S\ref{sec:main_burgers} for the forced case. The solution is first obtained by the method of characteristics, and then rederived using a perturbative approach, thereby establishing the validity of the perturbative framework.
Section \S\ref{sec:swe} extends and applies the perturbative framework to the shallow water equations with parabolic bathymetry and uses it to derive the GRW solutions presented in \S\ref{sec:main_swe}.
Section \S\ref{sec:applications} uses the one-sided GRW solution to describe the GRWs resulting from physical fluid configurations with finite mass over a parabolic well bathymetry. The parabolic well results are then extended to the case of a parabolic hill, in which the fluid is expelled outward rather than trapped within the well.
Section \S\ref{sec:conclusion} discusses the results and  future challenges. 

%% file: sections/3_burgers.tex
\section{Inviscid Burgers' equation}\label{sec:burgers}

We use the forced inviscid Burgers' equation~\eqref{eq:burgers-parabolic} with Riemann initial data~\eqref{eq:burgers-ic} as a benchmark problem for the perturbative framework developed below. This scalar problem provides a simpler setting in which the construction is derived and checked before being applied to the SWE system.

\subsection[Rarefaction solution of unforced Burgers' equation]{Unforced inviscid Burgers' equation}\label{sec:burgers-unforced}

We first recall the rarefaction wave (RW) solution of the unforced inviscid Burgers' equation, obtained from~\eqref{eq:burgers-parabolic} by setting \(\omega=0\):
\begin{equation}\label{eq:burgers-unforced}
    \frac{\partial u}{\partial t} + u \, \frac{\partial u}{\partial x} = 0.
\end{equation}
Equation~\eqref{eq:burgers-unforced} is invariant under the scaling
\begin{equation}\label{eq:burgers-scaling-invariance}
    (u,x,t) \mapsto (u, \lambda x, \lambda t), \qquad \lambda > 0.
\end{equation}
For centred Riemann data~\eqref{eq:burgers-ic}, which is invariant under this scaling, this suggests the evolution must be in the self-similar variable
\begin{equation}
    \xi=\frac{x}{t},
\end{equation}
and seek solutions of the form \(u(x,t)=U(\xi)\). 

Substitution into~\eqref{eq:burgers-unforced} yields
\begin{equation}
    (U-\xi)\,U'(\xi)=0.
\end{equation}
Hence, on each smooth self-similar branch, either \(U\) is constant or \(U(\xi)=\xi\). The constant branches correspond to the exterior states \(u_-\) and \(u_+\), while the linear branch represents the interior rarefaction fan.

For \(u_-<u_+\), the entropy solution of the centred Riemann problem is therefore the RW
\begin{equation}\label{eq:burgers-unforced-solution}
    u(x,t) =
    \begin{cases}
        u_{-}, & \xi<u_{-},\\[2pt]
        \xi,   & u_{-}<\xi<u_{+},\\[2pt]
        u_{+}, & \xi>u_{+},
    \end{cases}
    \qquad
    \xi=\frac{x}{t}.
\end{equation}

\begin{figure}[ht]
    \centering
    \includegraphics[width=0.95\linewidth]{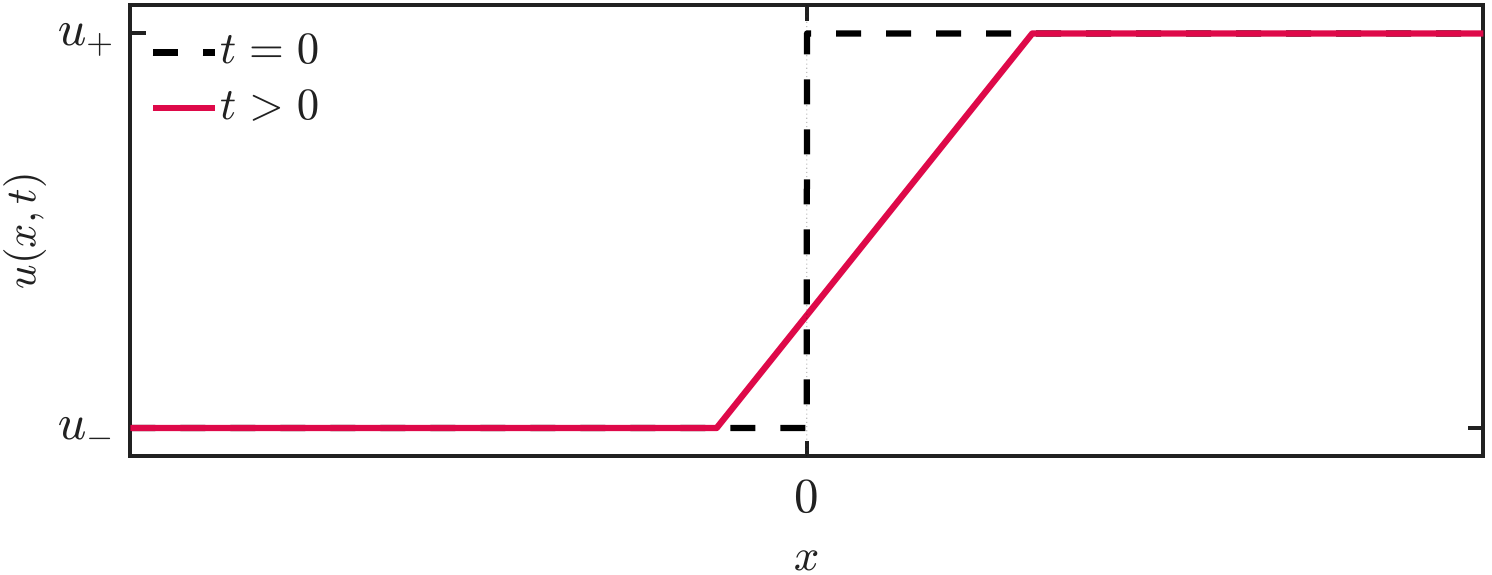}
    \caption{Evolution of the rarefaction wave
        \eqref{eq:burgers-unforced-solution} of the unforced Burgers'
        equation}
    \label{fig:burgers-unforced}
\end{figure}

\subsection{Forced inviscid Burgers' equation}\label{sec:burgers-parabolic}

We now consider the forced Burgers' equation~\eqref{eq:burgers-parabolic}. The parabolic forcing breaks the scaling symmetry~\eqref{eq:burgers-scaling-invariance} underlying the self-similar RW of the unforced problem, so the Riemann solution is no longer governed by the single variable \(\xi=x/t\).

This section develops a perturbative expansion for the GRW. We derive the GRW solution first using the the method of characteristics for centred Riemann data as a benchmark. The same solution is then derived using the perturbative expansion. The off-centred case then follows from Proposition~\ref{prop:translation}.

\subsubsection{Solution via the method of characteristics}\label{sec:burgers-moc}

Consider the forced Burgers' equation~\eqref{eq:burgers-parabolic} with
the Riemann data~\eqref{eq:burgers-ic} with $x_0=0$. Applying the method of characteristics yields the characteristic equations
\begin{equation}\label{eq:burgers-moc-char-odes}
\begin{aligned}
\frac{d\tilde x}{dt} = \tilde u(t), \qquad \frac{d\tilde u}{dt} = -\omega^{2}\tilde x (t),
\end{aligned}
\end{equation}
where \(\tilde x(t)\) is the characteristic curve emanating from
\(\tilde x(0) \doteq  \tilde x_{0}\) and \(\tilde u(t) \doteq u(\tilde x(t),t)\) is the value of the solution along this curve, with initial value
\(\tilde u_{0}\).
System~\eqref{eq:burgers-moc-char-odes} is a linear harmonic oscillator of frequency $\omega$, with solution
\begin{subequations}\label{eq:burgers-char_map}
    \begin{align}
        \tilde x(t) &= \tilde{x}_{0}\cos(\omega t)
            +\frac{\tilde u_{0}}{\omega}\sin(\omega t),
            \label{eq:burgers-char_map-x}\\[3pt]
        \tilde u(t) &= \tilde u_{0}\cos(\omega t)-\omega \tilde{x}_{0}\sin(\omega t).
            \label{eq:burgers-char_map-u}
    \end{align}
\end{subequations}
Inverting the characteristics~\eqref{eq:burgers-char_map} for \(\tilde x_{0}<0\) (where \(\tilde u_{0} = u_{-}\)), \(\tilde x_{0}>0\) (where \(\tilde u_{0} = u_{+}\)), and the rarefaction fan at \(\tilde x_{0}=0\) yields the explicit solution and edge expressions presented in \S\ref{sec:main_burgers}. The latter is recovered by substituting \(\tilde u_0 = u_{\pm}\) for \(\tilde{x}_0 = 0\) into~\eqref{eq:burgers-char_map-x}.

\begin{figure}[h!]
    \centering
    \includegraphics[width=\linewidth]{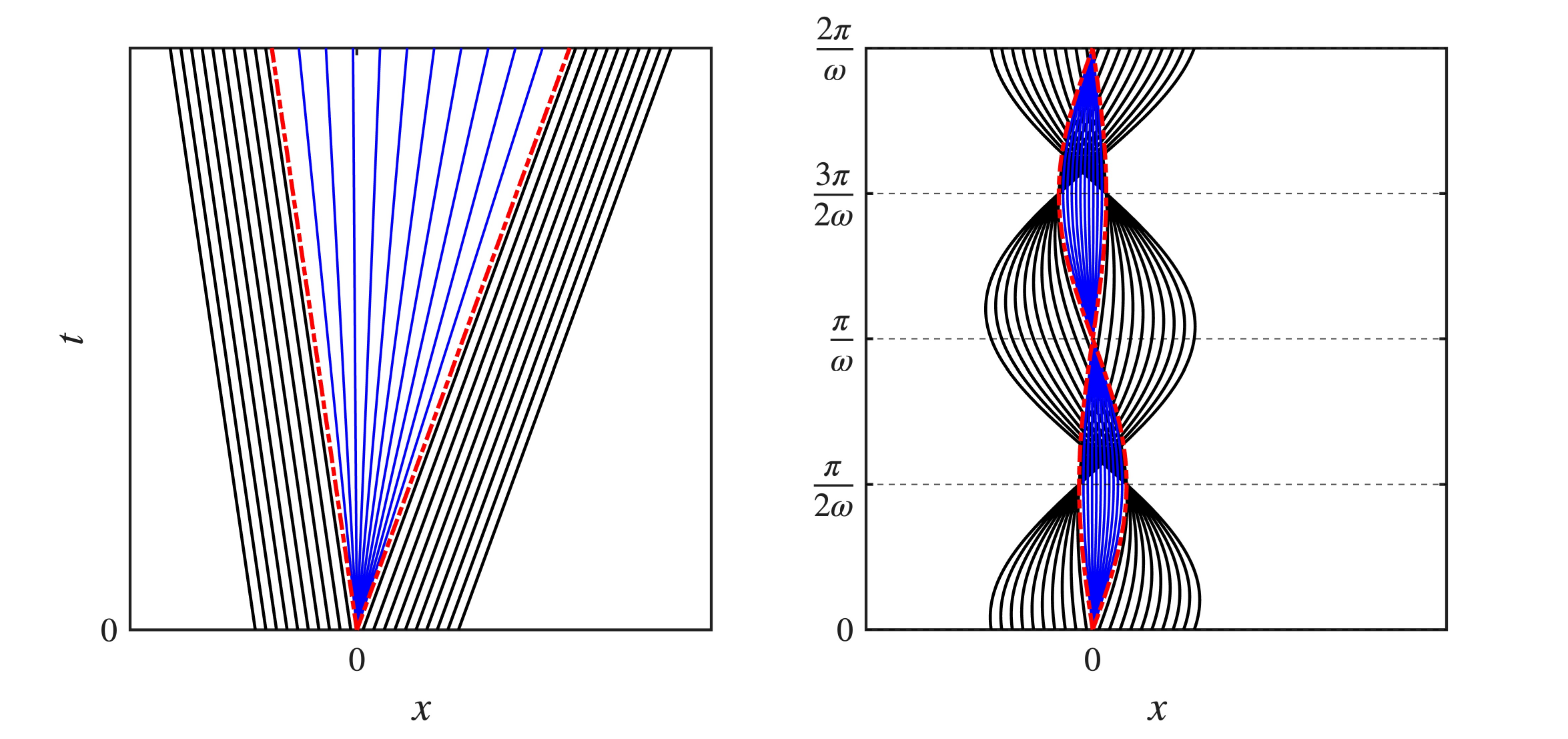}
    \caption{Characteristics for the unforced \textbf{(left)} and harmonically forced
    \textbf{(right)} Burgers equation with Riemann initial data~\eqref{eq:burgers-ic}.
    In the forced case, the characteristics outside the rarefaction fan focus and
    collapse periodically, with the first collapse at \(\omega t_{*}=\pi/2\) and
    subsequent collapses occurring at intervals of \(\pi/2\omega\).
    The red curves mark the rarefaction edges}
    \label{fig:burgers-characteristics}
\end{figure}

The characteristics~\eqref{eq:burgers-char_map} yield a well-defined solution only while the map \(x_{0}\mapsto \tilde x(t)\) remains one-to-one (see Figure~\ref{fig:burgers-characteristics}, right).
The Jacobian of this map, \(\partial\tilde x/\partial x_{0}=\cos(\omega t)\), vanishes at
\begin{equation}\label{eq:burgers-moc-singular-time}
    t_{*}=\frac{\pi}{2\omega}.
\end{equation}
As a result, the solution breaks down as \(\omega t\to\pi/2\). Importantly, this is not a standard shock as the solution does not develop finite left and right limits. The solution thus is valid in the interval
\begin{equation}\label{eq:burgers-moc-validity}
    0 < \omega t < \frac{\pi}{2}.
\end{equation}

Although the piecewise classical solution breaks down at \(\omega t=\pi/2\), the characteristic map~\eqref{eq:burgers-char_map} remains well defined for later times, providing the analytic continuation shown in the right panel of Figure~\ref{fig:burgers-characteristics}. For \(\omega t>\pi/2\), the two characteristic families periodically exchange roles such that the exterior characteristics defocus after their collapse at \(\omega t=\pi/2\), while the interior fan characteristics converge to a focal point at \(\omega t=\pi\). This alternating focusing--defocusing dynamics repeats at each multiple of \(\pi/2\), reflecting the underlying harmonic-oscillator dynamics~\eqref{eq:burgers-moc-char-odes}.

\subsubsection{Solution via a perturbative expansion}\label{sec:burgers-perturbation}

We now recover the solution~\eqref{eq:burgers-sol-main} through a
perturbative framework. 
Since the forcing is invariant under \(\omega\leftrightarrow -\omega\), we seek a regular perturbative expansion in even powers of \(\omega\),
\begin{equation}\label{eq:burgers-series-expansion}
    u(x,t;\omega)=\sum_{n=0}^{\infty} u_{n}(x,t)\,\omega^{2n}
    = u_{0}(x,t)+u_{1}(x,t)\,\omega^{2}+u_{2}(x,t)\,\omega^{4}+\cdots.
\end{equation}
The leading-order term \(u_{0}(x,t)\) is the unforced
solution~\eqref{eq:burgers-unforced-solution} with \(x_0=0\), while the
higher-order terms \(u_n(x,t)\), for \(n\geq 1\), encode the effect of the
forcing. Because the Riemann data is independent of \(\omega\), these
corrections for \(n\geq 1\) satisfy the homogeneous initial conditions
\begin{equation}\label{eq:burgers-corrections-IC}
    u_{n}(x,0)=0.
\end{equation}
This procedure carries out the perturbative construction separately in each region of the leading-order solution and then combines the resulting expressions to obtain the GRW~\eqref{eq:burgers-sol-main}.

\paragraph{\textbf{Interior solution}}

The leading order solution in the interior region is  \(u_{0}(x,t)=x/t\).
We show in Appendix~\ref{apdx:burgers-interior} that the higher-order corrections satisfy a hierarchy of forced linear first-order hyperbolic PDEs, which can be solved explicitly to all orders for \(u_n(x,t)\). The result is summarised as follows.

\begin{adjustwidth}{1.5em}{0pt}
\begin{proposition}[\textbf{Interior solution through perturbative expansion}]
\label{prop:burgers-interior}

\itshape
The perturbative expansion for the
leading-order solution \(u_{0}(x,t)=x/t\) is given by
\begin{equation}\label{eq:burgers-interior-solution-form}
u(x,t;\omega)=\frac{x}{t}\,F(\omega t),
\end{equation}
where \(F\) has the power series
\begin{equation}\label{eq:burgers-interior-series}
F(\omega t)=\sum_{n=0}^{\infty}F_{n}(\omega t)^{2n},
\end{equation}
with \(F_0 = 1\), \(F_1 = -1/3\) and the higher coefficients
for \(n\geq 2\) are determined through the recursion 
\begin{equation}\label{eq:burgers-interior-recursion}
F_{n}=-\frac{1}{2n+1}\sum_{j=1}^{n-1}F_{j}\,F_{n-j}~.
\end{equation}
Furthermore, \(F\) is the solution to the IVP
\begin{align}
\label{eq:burgers-interior-ode-F}
    t\,\frac{dF}{dt}-F+F^{2} &= -(\omega t)^{2}, \qquad F(0) = 1.
\end{align}

The solution of this IVP is
\begin{equation}\label{eq:burgers-interior-F-closed}
    F(\omega t)=\omega t\cot(\omega t).
\end{equation}
Substituting \(F(\omega t)\) into~\eqref{eq:burgers-interior-solution-form}
yields the interior piece of the GRW in~\eqref{eq:burgers-sol-main}.
\end{proposition}
\end{adjustwidth}

\paragraph{\textbf{Exterior regions}.}
In each exterior region, the leading-order solution is a constant state, \(u_{0}(x,t)=u_{\pm}\). In Appendix~\ref{apdx:burgers-exterior}, we show that the corresponding perturbation series can also be computed explicitly to all orders. The result is summarised as follows.

\begin{adjustwidth}{1.5em}{0pt}
\begin{proposition}[\textbf{Exterior solution through perturbative expansion}]
\label{prop:burgers-exterior}

\itshape
The perturbative expansion for a constant state \(u_{0}(x,t)=c\) for \(c\in \mathbb{R}\) is given by
\begin{equation}\label{eq:burgers-exterior-solution-form}
    u(x,t;\omega)=\frac{x}{t}\, F(\omega t)+c\,G(\omega t),
\end{equation}
where the  functions \( F\) and \(G\) have the power series
\begin{align}\label{eq:burgers-exterior-series}
    F(\omega t) = \sum_{n=0}^{\infty} F_{n}(\omega t)^{2n},\qquad 
    G(\omega t) = \sum_{n=0}^{\infty}G_{n}(\omega t)^{2n},
\end{align}
with initial coefficients
\(( F_{0},G_{0})=(0,1)\) and
\(( F_{1},G_{1})=(-1,\tfrac{1}{2})\),
and the higher coefficients for \(n\geq 2\) are given by the recursion
\begin{subequations}
\label{eq:burgers-exterior-series-recursion}
\begin{align}
    F_{n} = -\frac{1}{2n-1}\sum_{j=1}^{n-1} F_{j}\, F_{n-j},\qquad 
    G_{n} = -\frac{1}{2n}\sum_{j=1}^{n} F_{j}\,G_{n-j}.
\end{align}
\end{subequations}
Furthermore, the \(F\) and \(G\) are solutions to the IVP
\begin{subequations}
    \label{eq:burgers-exterior-odes}
    \begin{align}
    \label{eq:burgers-exterior-F}
        t\,\frac{d F}{dt}- F+ F^{2} &= -(\omega t)^{2},\\[2mm]
          \label{eq:burgers-exterior-G}
        t\,\frac{dG}{dt}+ F\,G &= 0,
    \end{align}
\end{subequations}
with initial data
\begin{subequations}
\label{eq:burgers-exterior-ode-IC}
\begin{align}
    F(0)=0, \qquad 
    G(0)=1.
\end{align}
\end{subequations}
The solution of this IVP is
\begin{subequations}
\begin{align}\label{eq:burgers-exterior-closed}
     F(\omega t) &=-\omega t\tan(\omega t), \\[2mm]
    \qquad
    G(\omega t) &=\sec(\omega t).
\end{align}
\end{subequations}
Substituting $F(\omega t)$ and $G(\omega t)$ into~\eqref{eq:burgers-exterior-solution-form} with \(c=u_{\pm}\)  yields the exterior pieces of the GRW in~\eqref{eq:burgers-sol-main}.
\end{proposition}
\end{adjustwidth}

\paragraph{\textbf{Rarefaction edges}.} 
To complete the construction of the GRW solution, we determine the rarefaction edges \(X_\pm(t;\omega)\). Each edge separates the rarefaction fan from the adjacent exterior region of constant state \(u_\pm\), and is therefore the limiting characteristic of that exterior region (see Figure~\ref{fig:burgers-characteristics}).
Every characteristic curve \(X(t)\) of the forced Burgers' equation is carried along at the local characteristic velocity $u$, and therefore
\begin{equation}\label{eq:char_vel_burgers}
  \frac{dX}{dt} = u(X(t),t).
\end{equation}
The two edges are the characteristics singled out by the initial discontinuity, originating at \(x_0 = 0\) and carrying the exterior data \(u_\pm\). Evaluating the right-hand side of~\eqref{eq:char_vel_burgers} on the exterior side using the exterior solution from Proposition~\ref{prop:burgers-exterior} yields the IVP
\begin{equation}\label{eq:burgers-edge-ivp}
    \frac{dX_\pm}{dt} + \omega\tan(\omega t)\,X_\pm(t)
    = u_\pm \sec(\omega t),
    \qquad X_\pm(0) = 0,
\end{equation}
whose solution,
\begin{equation}
    X_\pm(t;\omega) = \frac{u_\pm}{\omega}\sin(\omega t),
\end{equation}
agrees with~\eqref{eq:burgers-edges} presented in the main results section. This completes the GRW construction.

\subsubsection{A single system for the Riemann solution}
\label{sec:burgers-ode-structure}

The interior and exterior ODE systems obtained in Propositions~\ref{prop:burgers-interior} and~\ref{prop:burgers-exterior} can be viewed as two cases of a single underlying system, with each branch of the Riemann solution distinguished only by its initial data. Both the interior and exterior solutions have the general form
\begin{equation}\label{eq:burgers-general-form}
    u(x,t;\omega) = \frac{x}{t}\,F(\omega t) + c\,G(\omega t),
\end{equation}
where $F$ and $G$ satisfy
\begin{subequations}\label{eq:burgers-common-ode-system}
\begin{align}
    t\,\frac{dF}{dt}-F+F^{2} &= -(\omega t)^{2},\\[2mm]
    t\,\frac{dG}{dt}+F\,G &= 0.
\end{align}
\end{subequations}
For the interior, the initial data is $(F,G)(0) = (1,0)$. Since $G(0) = 0$ and the $G$ equation is homogeneous in $G$, it follows that $G(t) \equiv 0$, reducing~\eqref{eq:burgers-general-form} to $u = (x/t)\,F$, which is the interior solution of Proposition~\ref{prop:burgers-interior}. For the exterior with constant state $c = u_\pm$, the initial data is $(F,G)(0) = (0,1)$, giving the full two-component solution of Proposition~\ref{prop:burgers-exterior}. The same unified ODE structure appears in Section~\S\ref{sec:swe} for the forced SWE which is discussed next.

%% file: sections/4_swe.tex
\section{Shallow water equations}\label{sec:swe}
We now extend the perturbative framework of Section~\S\ref{sec:burgers} to the forced SWE~\eqref{eq:swe-parabolic}. After recalling the flat-bathymetry dam-break solution, we reduce the forced problem to an ODE system for four time-dependent coefficients, whose solutions give the various branches of the Riemann solution.

RW solutions of the SWE are most conveniently described in terms of Riemann variables (RVs), which diagonalise the system and cast it in characteristic form. We define
\begin{equation}\label{eq:RV_def}
    R_{\pm}(x,t) \;=\; u(x,t) \pm 2\sqrt{h(x,t)}.
\end{equation}
In terms of $R_{\pm}$, the forced SWE~\eqref{eq:swe-parabolic} take the diagonal form
\begin{equation}\label{eq:diagonal_SWE}
    \frac{\partial R_{\pm}}{\partial t} \;+\; \lambda_{\pm}\,\frac{\partial R_{\pm}}{\partial x}
    \;=\; -\omega^{2} x, \qquad \text{with} \quad \lambda_{\pm} \;=\; \frac{3R_{\pm}+R_{\mp}}{4} \;=\; u \pm \sqrt{h},
\end{equation}
where \(\lambda_{\pm}\) are the characteristic velocities.
When $\omega=0$ the forcing vanishes and the $R_{\pm}$ reduce to the classical
Riemann invariants of the unforced system, each constant along its own characteristic family.

\subsection{Flat-bathymetry dambreak solution}\label{sec:stoker_solution}

For $\omega=0$, the SWE~\eqref{eq:swe-parabolic} with dam-break initial
data~\eqref{eq:swe-dambreak-ic} admit the classical self-similar dam-break solution, which is referred to in this work as the \emph{Stoker solution}~\cite{stoker2019water}. Writing $\xi = x/t$ and $c_0 = \sqrt{h_0}$,
\begin{subequations}\label{eq:stoker_solution}
\begin{align}
    h(x,t) &=
    \begin{cases}
        h_0,                                & \xi < -c_0,\\[1mm]
        \tfrac{1}{9}(2c_0-\xi)^{2},        & -c_0 \le \xi \le 2c_0,\\[1mm]
        0,                                  & \xi > 2c_0,
    \end{cases}
    \\[2mm]
    u(x,t) &=
    \begin{cases}
        0,                                  & \xi < -c_0,\\[1mm]
        \tfrac{2}{3}(\xi+c_0),             & -c_0 \le \xi \le 2c_0,\\[1mm]
        \xi,                                & \xi > 2c_0.
    \end{cases}
\end{align}
\end{subequations}
In the vacuum region $\xi > 2c_0$ the depth vanishes and we take $u = \xi$
throughout. For a rigorous discussion of the vacuum state see~\cite{liu1980vacuum}.

Expressed in terms of the RVs, the Stoker solution reads
\begin{subequations}\label{eq:stoker_solution_RI}
\begin{align}
    R_{+}(x,t) &=
        \begin{cases}
            2c_0, & \xi < -c_0,\\[1pt]
            2c_0, & -c_0 \le \xi \le 2c_0,\\[1pt]
            \xi,  & \xi > 2c_0,
        \end{cases}
    \\[2mm]
    R_{-}(x,t) &=
        \begin{cases}
            -2c_0,                                  & \xi < -c_0,\\[1pt]
            -\tfrac{2}{3}c_0 + \tfrac{4}{3}\xi,    & -c_0 \le \xi \le 2c_0,\\[1pt]
            \phantom{+}\xi,                         & \xi > 2c_0.
        \end{cases}
\end{align}
\end{subequations}
Across the rarefaction fan $R_{+} = 2c_0$ is constant while $R_{-}$ varies
linearly with $\xi$ (see Figure~\ref{fig:stoker-combined}). We refer to this as the
\emph{slow simple-wave} rarefaction, which is relevant to the rightward expanding dam-break data~\eqref{eq:swe-dambreak-ic}. The complementary
\emph{fast simple-wave} rarefaction, in which $R_{-}$ is constant and $R_{+}$ varies,
corresponds to the mirrored configuration with vacuum on the left and $h_0$ on
the right, and is obtained from the slow rarefaction by flipping the sign
$c_0 \mapsto -c_0$.

\begin{figure}[h!]
    \centering
    \includegraphics[width=\linewidth]{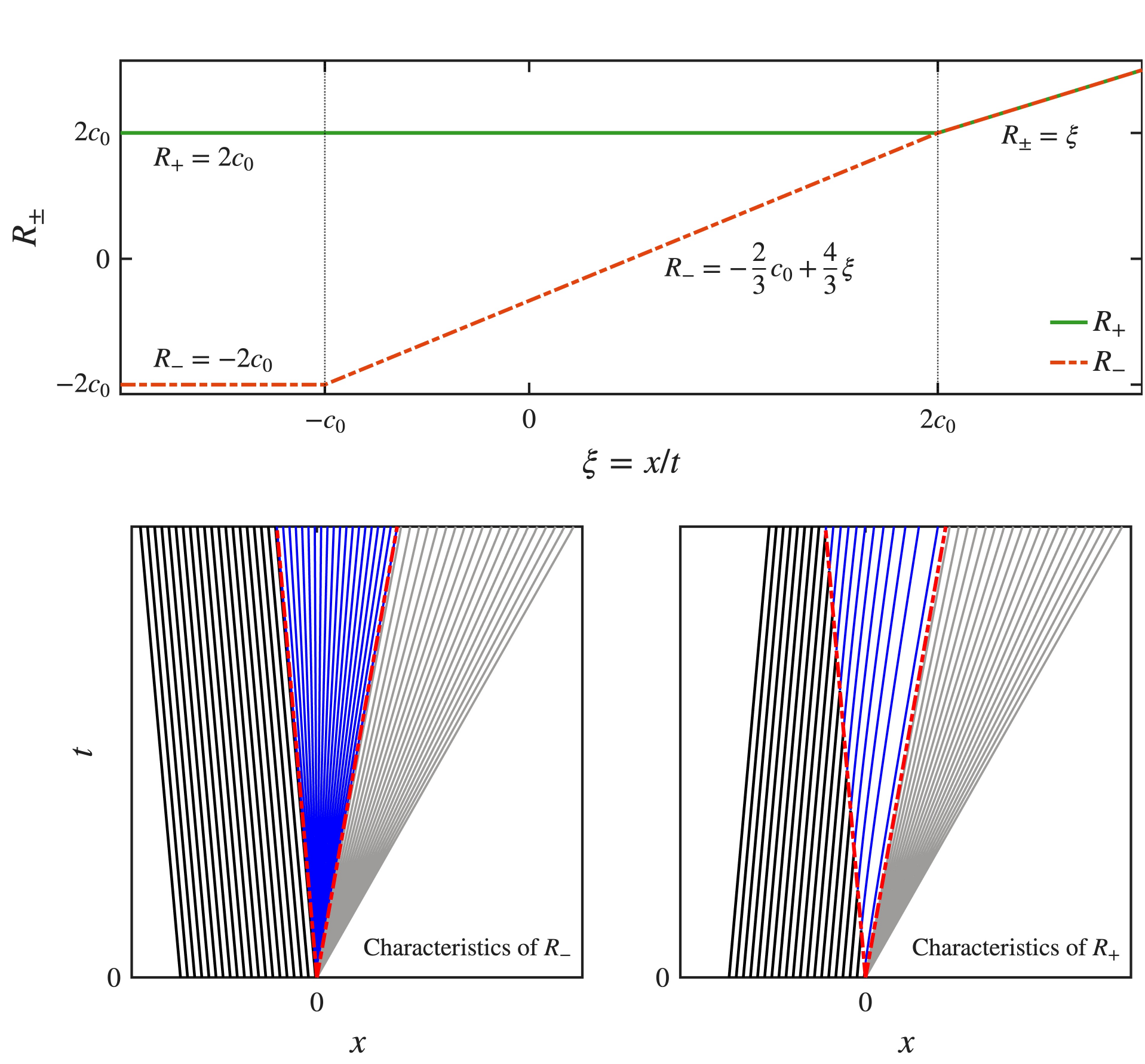}
    \caption{Flat-bathymetry dam-break solution~\eqref{eq:stoker_solution_RI}. 
    The top panel shows the Riemann variables $R_\pm$ as functions of the self-similar variable $\xi=x/t$. Across the rarefaction fan, $R_+$ remains constant, while $R_-$ varies linearly. 
    The bottom panels show the corresponding characteristic curves for $R_-$ (bottom left) and $R_+$ (bottom right), with the red dashed curves marking the rarefaction edges $x=-c_0 t$ and $x=2c_0 t$}
    \label{fig:stoker-combined}
\end{figure}
\clearpage

\subsection{Perturbative construction of the generalised rarefaction solution}\label{sec:swe_perturbation}

Treating the parabolic forcing as a perturbation of the flat-bathymetry Stoker
solution~\eqref{eq:stoker_solution_RI}, we expand the RVs in a regular perturbative expansion in 
even powers of $\omega$,
\begin{equation}\label{eq:SWE_expansion}
    R_{\pm}(x,t;\omega) \;=\; \sum_{n=0}^{\infty} r^{\pm}_{n}(x,t)\,\omega^{2n},
\end{equation}
where $r^{\pm}_{0}$ is taken from a branch of the Stoker
solution~\eqref{eq:stoker_solution_RI} and $r^{\pm}_{n}$ for $n\geq 1$ satisfies
homogeneous initial data. As in the Burgers case
(\S\ref{sec:burgers-ode-structure}), the perturbative calculation yields a
single ODE system governing all branches of the Riemann solution, with each branch
selected by its admissible initial data.

\begin{adjustwidth}{1.5em}{0pt}
\begin{proposition}[\textbf{ODE system for Riemann solution branches}]\label{prop:swe-series-ode}
\itshape
For each branch of the Stoker solution~\eqref{eq:stoker_solution_RI}, the
perturbative construction~\eqref{eq:SWE_expansion} yields a solution of the form
\begin{equation}\label{eq:affine_form}
    R_{\pm}(x,t;\omega) = 2c_0\,A_{\pm}(\omega t) + \frac{x}{t}\,B_{\pm}(\omega t),
\end{equation}
where the time-dependent coefficients $A_{\pm}$ and $B_{\pm}$ satisfy
\begin{subequations}\label{eq:ABCD_ODE_swe}
\begin{align}
    t\,\frac{dA_{\pm}}{dt}
        + \frac{1}{4}\bigl(3A_{\pm}+A_{\mp}\bigr)B_{\pm}
        &= 0,\label{eq:ABCD_ODE_swe_Apm}\\[1mm]
    t\,\frac{dB_{\pm}}{dt}
        - B_{\pm}
        + \frac{1}{4}\bigl(3B_{\pm}+B_{\mp}\bigr)B_{\pm}
        &= -(\omega t)^{2}.\label{eq:ABCD_ODE_swe_Bpm}
\end{align}
\end{subequations}
The initial data select the rarefaction, uniform depth,
or vacuum branches of the Riemann solution:
\begin{equation}\label{eq:ABCD_branches}
    \bigl(A_{+},A_{-},B_{+},B_{-}\bigr)(0)
    =
    \begin{cases}
        (1,\,-\tfrac{1}{3},\,0,\,\tfrac{4}{3}),    & \text{slow simple wave rarefaction,}\\[1mm]
        (-\tfrac{1}{3},\,1,\,\tfrac{4}{3},\,0),    & \text{fast simple wave rarefaction,}\\[1mm]
        (1,\,-1,\,0,\,0),                           & \text{uniform depth,}\\[1mm]
        (0,\,0,\,1,\,1),                            & \text{vacuum.}
    \end{cases}
\end{equation}
\end{proposition}
\end{adjustwidth}
Appendix~\ref{apdx:swe} contains the derivation of Proposition~\ref{prop:swe-series-ode} and the perturbative calculation.

\subsubsection{Invariant manifolds and solution branches}\label{sec:swe_branches}

The ODE system~\eqref{eq:ABCD_ODE_swe} is four-dimensional, but the solution branches associated with the Riemann problem lie on two invariant manifolds of the \((B_{+},B_{-})\) subsystem. On each manifold, the coupled equations for \(B_{+}\) and \(B_{-}\) reduce to a scalar Riccati equation. The equations for \(A_{\pm}\) then become linear and can be integrated using an integrating factor. The two invariant manifolds are
\begin{subequations}\label{eq:invariant_manifolds}
\begin{align}
    B_{+}(\omega t) &= B_{-}(\omega t),
    \label{eq:uniform_condition}\\[1mm]
    B_{+}(\omega t)\,B_{-}(\omega t) &= -(\omega t)^{2}.
    \label{eq:RW_condition}
\end{align}
\end{subequations}
The first manifold contains the uniform-depth and vacuum branches, while the second contains the slow and fast simple-wave branches. The invariance of these manifolds under~\eqref{eq:ABCD_ODE_swe} is established in Appendix~\ref{apdx_sec:BD_system}.

\paragraph{\textbf{Uniform depth and vacuum branches}}\label{sec:swe_spatially_constant}
On the manifold~\eqref{eq:uniform_condition}, the $B_{\pm}$ equations
in~\eqref{eq:ABCD_ODE_swe} reduce to the scalar Riccati equation
\begin{equation}\label{eq:BD_uniform_ode}
    t\, \frac{d B_{+}}{dt} - B_{+} + B_{+}^{2} = -(\omega t)^{2}, \qquad B_{-} = B_{+}.
\end{equation}

\begin{itemize}
    \item \textbf{Uniform depth.}
    With initial data $(A_{+},A_{-},B_{+},B_{-})(0)=(1,-1,0,0)$
    from~\eqref{eq:ABCD_branches}, the solution is
    \begin{equation}\label{eq:BD_uniform_closed}
        B_{+} = B_{-} = -\omega t\tan(\omega t), \qquad
        A_{+} = -A_{-} = \sqrt{\sec(\omega t)},
    \end{equation}
    Substituting into the affine solution form~\eqref{eq:affine_form} for
    $R_{\pm}$ and inverting the RVs~\eqref{eq:RV_def} yields the depth and velocity
    \begin{equation}\label{eq:swe_uniform_phys}
        h(x,t) = h_0\sec(\omega t), \qquad u(x,t) = -\omega x\tan(\omega t).
    \end{equation}

    \item \textbf{Vacuum state.}
    With initial data $(A_{+},A_{-},B_{+},B_{-})(0)=(0,0,1,1)$
    from~\eqref{eq:ABCD_branches}, the solution is
    \begin{equation}\label{eq:BD_vacuum_closed}
        B_{+} = B_{-} = \omega t\cot(\omega t), \qquad
        A_{+} = A_{-} = 0,
    \end{equation}
    Substituting into~\eqref{eq:affine_form} and inverting~\eqref{eq:RV_def} yields the depth and velocity
    \begin{equation}\label{eq:swe_vacuum_phys}
        h(x,t) = 0, \qquad u(x,t) = \omega x\cot(\omega t).
    \end{equation}
\end{itemize}

\paragraph{\textbf{Simple wave branches and the run-up solution}}\label{sec:swe_rw_solutions}
On the manifold~\eqref{eq:RW_condition}, the $B_{\pm}$ equations 
in~\eqref{eq:ABCD_ODE_swe} reduce to the scalar Riccati equation
\begin{equation}\label{eq:D_RW_ode}
    t\, \frac{d B_{+}}{dt} - B_{+} + \tfrac{3}{4}B_{+}^{2} = -\tfrac{3}{4}(\omega t)^{2}, \qquad B_{-} = -\frac{(\omega t)^{2}}{B_{+}}
\end{equation}

\begin{itemize}
    \item \textbf{Slow simple wave.}
    With initial data $(A_{+},A_{-},B_{+},B_{-})(0)=(1,-\tfrac{1}{3},0,\tfrac{4}{3})$
    from~\eqref{eq:ABCD_branches}, the solution is
    \begin{equation}\label{eq:ABCD_slow_branch}
        \begin{aligned}
        B_{+}(\omega t) &= -\omega t\tan\bigl(\tfrac{3\omega t}{4}\bigr), \qquad
        A_{+}(\omega t) = \frac{\cos(\frac{\omega t}{4})}{\cos(\frac{3\omega t}{4})},\\[2mm]
        B_{-}(\omega t) &= \phantom{-}\omega t\cot\bigl(\tfrac{3\omega t}{4}\bigr), \qquad
        A_{-}(\omega t) = -\frac{\sin(\frac{\omega t}{4})}{\sin(\frac{3\omega t}{4})}.
        \end{aligned}
    \end{equation}
    Substituting into~\eqref{eq:affine_form} and inverting~\eqref{eq:RV_def} yields the generalised slow simple wave branch or the \emph{run-up} solution
    \begin{subequations}\label{eq:swe_slow_phys}
    \begin{align}
        \hrunup(x,t;\omega)
        &= \frac{\omega^{2}}{4}\csc^{2}\!\left(\frac{3\omega t}{2}\right)
            \!\left(\frac{2c_0}{\omega}\sin(\omega t)-x\right)^{2},
            \label{eq:swe_slow_phys_h}
            \\[2mm]
        \urunup(x,t;\omega)
        &= \frac{2c_0 + \omega x\,\cos(\tfrac{3\omega t}{2})
            \csc(\tfrac{\omega t}{2})}{1+2\cos(\omega t)}.
            \label{eq:swe_slow_phys_u}
    \end{align}
    \end{subequations}

    \item \textbf{Fast simple wave.}
    The generalised fast simple wave branch, corresponding to the mirrored configuration with vacuum on the left, can be obtained from~\eqref{eq:swe_slow_phys} by the flipping the sign $c_0\mapsto -c_0$.
\end{itemize}

\paragraph{\textbf{Asymptotic nature of the run-up solution.}}
The pair~\eqref{eq:swe_slow_phys} is an exact solution of the forced
SWE~\eqref{eq:swe-parabolic} and generalises the simple wave branch of the
classical Stoker rarefaction solution to parabolic bathymetry. In the flat
bathymetry limit $\omega\to 0$, it extends across the entire rarefaction fan and
recovers Stoker's solution globally. For $\omega>0$, however, it describes the
rarefaction flow only locally near the vacuum edge. A wave front expansion about
$x=X_{\rm dry}(t;\omega)$ shows that this run-up solution appears as the solution to the
leading order system in the resulting hierarchy. A full description of the nontrivial rarefaction
fan profile, denoted $H_{\rm fan}$ and $U_{\rm fan}$ in
Section~\S\ref{sec:main_results}, requires including the higher order spatial terms, which
are governed by a hierarchy of linear ODE systems whose dependence on $\omega$
is non-analytic and include fractional powers of $\omega$. The regular
perturbative expansion~\eqref{eq:SWE_expansion} thus captures only the analytic
component of the full solution. Further details are given in
Appendix~\ref{apdx_sec:wavefront_higher_order}.

\subsection{Rarefaction edges}\label{sec:swe_edges}

The rarefaction fan is bounded by two characteristics of the slow ($\lambda_{-}$) family (see Figure~\ref{fig:stoker-combined}). The dry edge $X_{\rm dry}(t)$ is the continuous boundary separating the fan from the vacuum, and the wet edge $X_{\rm wet}(t)$ is the continuous boundary separating the fan from the uniform-depth region. Since the depth and velocity are continuous across each edge, the bounding characteristic is shared by the fan and the adjacent exterior state, and the characteristic speed at the edge takes the same value whether computed from inside the fan or from the exterior. The exterior states are known exactly through the closed forms~\eqref{eq:swe_uniform_phys} and~\eqref{eq:swe_vacuum_phys}, whereas the fan interior lacks a complete description (see Figure~\ref{fig:swe-parabolic-solution}), so we locate each edge by evaluating its speed from the exterior side. This gives an IVP for each edge location. Both edges satisfy
\begin{equation}\label{eq:edge_speed}
    \frac{\mathrm{d}X}{\mathrm{d}t} = \lambda_{-}\bigl(R_{+}(X(t),t), \,R_{-}(X(t),t)\bigr) = u\big(X(t),t\big) - \sqrt{h\big(X(t),t\big)} ,
\end{equation}
with initial locations set such that $X_{\rm dry}(0)=X_{\rm wet}(0)=0$.

\paragraph{\textbf{Dry edge.}}

On the vacuum branch~\eqref{eq:swe_vacuum_phys}, $h \equiv 0$ so $\lambda_- = u = \omega x\cot(\omega t)$, and~\eqref{eq:edge_speed} becomes
\begin{equation}\label{eq:edge_dry_ivp}
    \frac{\mathrm{d}X_{\rm dry}}{\mathrm{d}t} - \omega\cot(\omega t)\,X_{\rm dry} = 0,
    \qquad X_{\rm dry}(0) = 0,
\end{equation}
with solution
\begin{equation}\label{eq:swe_edge_right}
    X_{\rm dry}(t;\omega) = \frac{2c_0}{\omega}\sin(\omega t).
\end{equation}

\paragraph{\textbf{Wet edge.}}
On the uniform-depth branch~\eqref{eq:swe_uniform_phys}, $u = -\omega x\tan(\omega t)$ and $h = h_0\sec(\omega t)$, so~\eqref{eq:edge_speed} gives
\begin{equation}\label{eq:edge_wet_ivp}
    \frac{\mathrm{d}X_{\rm wet}}{\mathrm{d}t} + \omega\tan(\omega t)\,X_{\rm wet} = -c_0\sqrt{\sec(\omega t)},
    \qquad X_{\rm wet}(0) = 0.
\end{equation}
Multiplying by the integrating factor $\sec(\omega t)$ and integrating gives
\begin{equation}\label{eq:swe_edge_left}
    X_{\rm wet}(t;\omega) = -\frac{c_0}{\omega}\cos(\omega t)\!\int_{0}^{\omega t}\!\sec^{3/2}(z)\,\mathrm{d}z.
\end{equation}
The integral in~\eqref{eq:swe_edge_left} has a closed-form expression in terms of the incomplete elliptic integral of the second kind with elliptic parameter \(m=2\),
\begin{equation}\label{eq:swe_edge_left_elliptic}
    \int_{0}^{\omega t}\!\sec^{3/2}(z)\,\mathrm{d}z = 2\sqrt{\sec(\omega t)}\sin(\omega t) - 2\,E\!\left(\frac{\omega t}{2},\,2\right).
\end{equation}

%% file: sections/5_applications.tex
\section{Applications to physical configurations}\label{sec:applications}

In this section we extend and apply the one-sided Riemann solution of Section~\S\ref{sec:swe} to physical fluid configurations with compact support and  finite mass. We first consider fluid configurations over a parabolic well, namely a box configuration (\S\ref{sec:swe_hat}) and a lake-at-rest configuration (\S\ref{sec:swe_stationary_profile}), and then extend the parabolic well results to a parabolic hill bathymetry (\S\ref{sec:swe_repulsive}).
We show that the finite-time singularity present in the GRW solution of one-sided Riemann data is never reached.

\subsection{Box-type configuration}\label{sec:swe_hat}

We consider the box-shaped depth profile, a constant-depth \emph{plateau} of height $h_0$ and width $x_{0}^{R}-x_{0}^{L}$ over a dry bed,
\begin{equation}\label{eq:swe-hat_data}
    h(x,0) =
    \begin{cases}
        h_0, & x_{0}^{L} < x < x_{0}^{R},\\[1mm]
        0, & \text{elsewhere,}
    \end{cases}
    \qquad u(x,0) = 0.
\end{equation}

\begin{figure}[b!]
    \centering
    \includegraphics[width=0.99\linewidth]{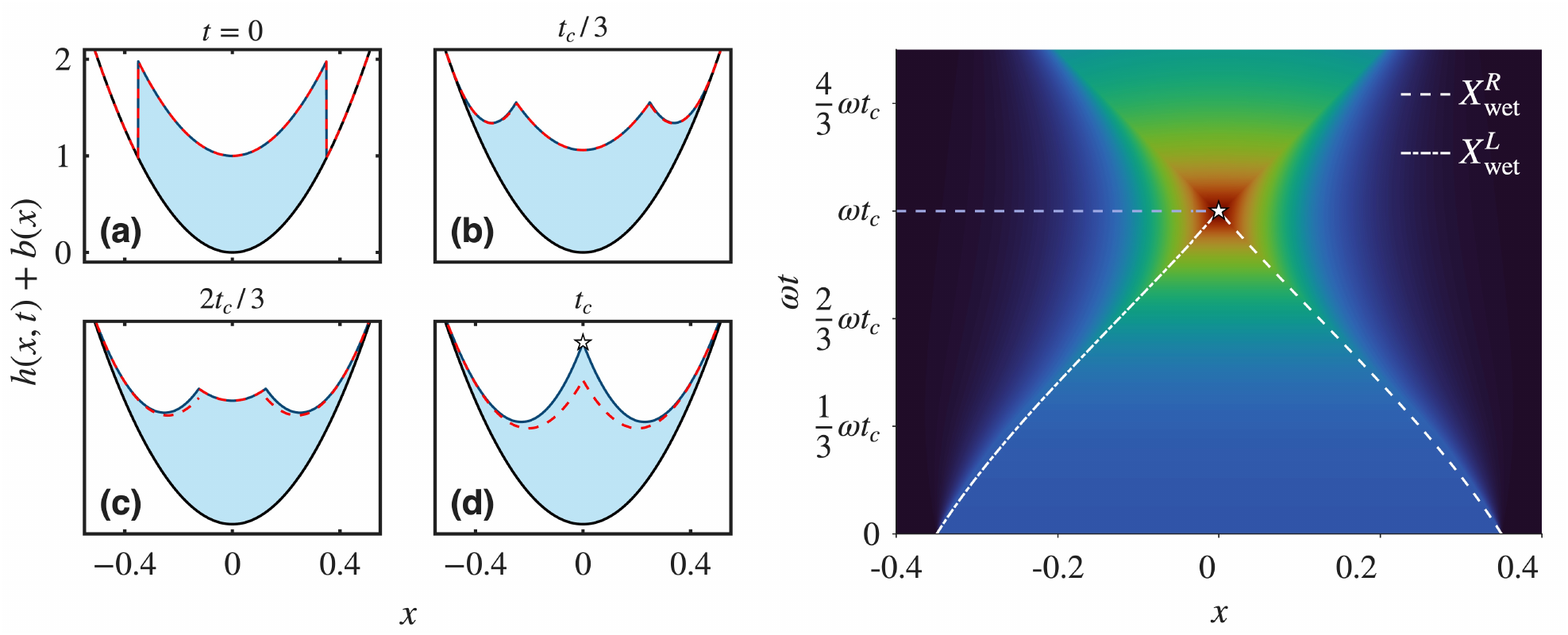}
    \caption{Evolution of the fluid depth starting from a box-type configuration over a parabolic bathymetry. The left panels show depth profiles at selected times, while the right panel shows a space-time colour map of $h(x,t)$ in the $(x,\omega t)$-plane. The piecewise solution~\eqref{eq:h_hat_sol} describes the evolution up to $t=t_c$, when the two peaks (inner edges) collide. The red dashed curves show the exact branches and the asymptotic run-up branch. At $t=t_c$ a peak forms with depth \(h_{\rm peak} = h_0 \sec(\omega t_c)\), marked by a star, which for $t>t_c$ splits into two outward-moving peaks.}
    \label{fig:hat_ic_evolution_panels}
\end{figure}

The left jump evolves as a left-moving generalised slow simple wave and the right jump evolves as a right-moving generalised fast simple wave. 
For each fan we label its dry edge (vacuum side) as \emph{outer} and its wet edge (uniform-depth side) as \emph{inner}. Until the two fans interact at $t \doteq t_c$ (given by \eqref{eq:swe-collision_time_solve} and \eqref{eq:tc-series-flat}), the solution is
\begin{equation}\label{eq:h_hat_sol}
    (h,u)(x,t) =
    \begin{cases}
        \bigl(0,\,\omega x\cot (\omega t)\bigr),
            & x < X_{\rm dry}^{L}(t;\omega), \\[1mm]
        \bigl(H_{\rm fan}^{L},\,U_{\rm fan}^{L}\bigr),
            & X_{\rm dry}^{L}(t;\omega) < x < X_{\rm wet}^{L}(t;\omega), \\[1mm]
        \bigl(h_0\sec (\omega t),\,-\omega x\tan (\omega t)\bigr),
            & X_{\rm wet}^{L}(t;\omega) < x < X_{\rm wet}^{R}(t;\omega), \\[1mm]
        \bigl(H_{\rm fan}^{R},\,U_{\rm fan}^{R}\bigr),
            & X_{\rm wet}^{R}(t;\omega) < x < X_{\rm dry}^{R}(t;\omega), \\[1mm]
        \bigl(0,\,\omega x\cot (\omega t)\bigr),
            & x > X_{\rm dry}^{R}(t;\omega).
    \end{cases}
\end{equation}
where $H_{\rm fan}^{L,R}$ and $U_{\rm fan}^{L,R}$ are the non-trivial fan profiles.
The inner and outer edges are given by
\begin{equation}\label{eq:edges_hat}
\begin{aligned}
    X_{\rm dry}^{L}(t;\omega)  &= -\frac{2c_0}{\omega}\sin(\omega t) + x_{0}^{L}\cos(\omega t),
    & X_{\rm dry}^{R}(t;\omega) &= \phantom{-}\frac{2c_0}{\omega}\sin(\omega t) + x_{0}^{R}\cos(\omega t),\\[2mm]
    X_{\rm wet}^{L}(t;\omega)  &=\frac{c_0}{\omega}\,\cF(\omega t)\cos(\omega t) + x_{0}^{L}\cos(\omega t),
    & X_{\rm wet}^{R}(t;\omega) &= -\frac{c_0}{\omega}\,\cF(\omega t) \cos(\omega t) + x_{0}^{R}\cos(\omega t),
\end{aligned}
\end{equation}
where
\begin{equation}\label{eq:F_sec}
    \cF(z) = \int_{0}^{z}\sec^{3/2}(\theta)\,\mathrm{d}\theta.
\end{equation}

Near the outer edges, $H_{\rm fan}^{L}$ and $H_{\rm fan}^{R}$ have the run-up asymptotics
\begin{equation}\label{eq:hat-runup}
    H_{\rm fan}^{L/R}(x,t;\omega) \;\sim\; \frac{\omega^{2}}{4}\csc^{2}\!\left(\frac{3\omega t}{2}\right)\bigl(X_{\rm dry}^{L/R}(t;\omega)-x\bigr)^{2}
    \qquad \text{as } x \to X_{\rm dry}^{L/R}(t;\omega)^{-}.
\end{equation}

The solution~\eqref{eq:h_hat_sol} inherits the singular factors of the one-sided Riemann solution (\S\ref{sec:main_swe}). In the box configuration, however, these finite-time singularities are never reached, because the two inner edges invariably collide before the singularity time. After the collision, the counter-propagating fans interact and form a complex-wave region in which, in the corresponding unforced problem, both $R_+$ and $R_-$ vary; this regime has been studied using the hodograph transformation~\cite{jenssen2017exact,pang2012collision}. Extending the present construction to the post-collision regime over parabolic bathymetry lies outside the scope of this work.

\begin{adjustwidth}{1.5em}{0pt}
\begin{proposition}[\textbf{Inner-edge collision and singularity prevention}]\label{prop:hat-collision}
\itshape
The inner edges $X_{\rm wet}^{L}$ and $X_{\rm wet}^{R}$ collide at the location
\begin{equation}\label{eq:x_c}
    x_c = \frac{1}{2}(x_{0}^{R}+x_{0}^{L})\cos(\omega t_c),
\end{equation}
with collision time $t_c$ given implicitly by
\begin{equation}\label{eq:swe-collision_time_solve}
    \omega t_c = \cF^{-1}\!\left[\frac{\omega}{2c_0}(x_{0}^{R}-x_{0}^{L})\right] < \frac{\pi}{2}.
\end{equation}
In particular, the collision occurs strictly before the singularity time $\omega t = \pi/2$ for any plateau width $x_0^R - x_0^L > 0$, so the finite-time blow-up predicted by the one-sided Riemann solution is prevented.
\end{proposition}
\end{adjustwidth}

\medskip
\begin{adjustwidth}{1.5em}{0pt}
\begin{proof}
Equating the inner edges $X_{\rm wet}^{L}(t;\omega)=X_{\rm wet}^{R}(t;\omega)$ from~\eqref{eq:edges_hat} gives the collision location~\eqref{eq:x_c} together with the condition
\begin{equation*}
    \cF(\omega t_c) = \frac{\omega}{2c_0}\,(x_{0}^{R}-x_{0}^{L}) \;=:\; \mu > 0,
\end{equation*}
where $\mu>0$ for any plateau width $x_{0}^{R}-x_{0}^{L}>0$. The integrand $\sec^{3/2}\theta$ in~\eqref{eq:F_sec} is positive and continuous on $[0,\pi/2)$, so $\cF$ is strictly increasing there with $\cF(0)=0$ and $\cF(z)\to\infty$ as $z\to(\pi/2)^{-}$. Hence $\cF$ is a strictly increasing bijection from $[0,\pi/2)$ onto $[0,\infty)$, and the equation $\cF(\omega t_c)=\mu$ admits a unique solution $\omega t_c\in[0,\pi/2)$ for every $\mu>0$. The collision therefore always occurs, at a time $t_c$ given by~\eqref{eq:swe-collision_time_solve}, independently of the box configuration.
\end{proof}
\end{adjustwidth}

\medskip
\noindent Denoting $t_c^{\flat} = \tfrac{x_0^{R}-x_0^{L}}{2c_0}$ as the flat-bottom collision time and noting $\mu = \omega t_c^{\flat}$, inversion of~\eqref{eq:swe-collision_time_solve} gives
\begin{equation}\label{eq:tc-series-flat}
    t_c = t_c^{\flat}\left[\,1
    - \frac{1}{4}\,(\omega t_c^{\flat})^{2}
    + \frac{17}{160}\,(\omega t_c^{\flat})^{4}
    - \frac{143}{2688}\,(\omega t_c^{\flat})^{6}
    + \frac{11191}{387072}\,(\omega t_c^{\flat})^{8}
    + \cdots\right].
\end{equation}

\newpage
\subsection{Lake-at-rest configuration}\label{sec:swe_stationary_profile}

A physically natural initial condition is a quiescent fluid sitting in the parabolic well, with the free surface in exact hydrostatic balance with the bathymetry. 
We consider the initial profile
\begin{equation}
\label{eq:swe-stationary_data}
    h(x,0) =
    \begin{cases}
        h_{\rm st}(x), & -x_{\rm st} < x < x_0,\\[1mm]
        0, & \text{elsewhere},
    \end{cases}
    \qquad u(x,0) = 0~,
\end{equation}
where 
\begin{subequations}\label{eq:h_stationary_profile}
\begin{gather}
    h_{\rm st}(x) = h_0 - b(x)
    = h_0 - \frac{1}{2}\omega^{2}x^{2},
    \label{eq:h_stationary_profile_depth}
    \\[3mm]
    x_0 = \delta x_{\rm st},
    \qquad
    x_{\rm st} = \frac{\sqrt{2}c_0}{\omega},
    \qquad
    c_0 = \sqrt{h_0}~.
    \label{eq:h_stationary_profile_params}
\end{gather}
\end{subequations}
where $h_0>0$.
The initial depth profile is non-negative on 
$[-x_{\rm st},x_{\rm st}]$ and vanishes at the endpoints. 
The jump is placed at $x_0$, such that $\delta = x_0/x_{\rm st} \in (-1,1)$ is the normalised jump position.
Since the initial total fluid elevation in the wet region is $\eta(x) = h_{\rm st}(x) + b(x) = h_0$ is constant, this represents a lake at rest. Hence, $h_{\rm st}$ corresponds to a steady state of the forced SWE~\eqref{eq:swe-parabolic}. The only motion is generated at the discontinuity $x=x_0$, where a rarefaction fan is released into the vacuum while the fluid to its left remains in lake-at-rest balance until reached by the fan.

For as long as an undisturbed stationary region remains, the depth and velocity have the form
\begin{subequations}\label{eq:h_stationary_sol}
\begin{align}
    h(x,t) &=
    \begin{cases}
        h_{\rm st}(x), & -x_{\rm st} < x < X^{\rm lake}_{\rm wet}(t;\omega),\\[1mm]
        H_{\rm fan}^{\rm st}(x,t;\omega), & X^{\rm lake}_{\rm wet}(t;\omega) < x < X^{\rm lake}_{\rm dry}(t;\omega),\\[1mm]
        0, & x > X^{\rm lake}_{\rm dry}(t;\omega),
    \end{cases}
    \label{eq:h_stationary_sol_h}
    \\[4mm]
    u(x,t) &=
    \begin{cases}
        0, & -x_{\rm st} < x < X^{\rm lake}_{\rm wet}(t;\omega),\\[1mm]
        U_{\rm fan}^{\rm st}(x,t;\omega), & X^{\rm lake}_{\rm wet}(t;\omega) < x < X^{\rm lake}_{\rm dry}(t;\omega),\\[1mm]
        \omega x\cot\omega t-\omega x_0\csc\omega t, & x > X^{\rm lake}_{\rm dry}(t;\omega),
    \end{cases}
    \label{eq:h_stationary_sol_u}
\end{align}
\end{subequations}
where the vacuum velocity in~\eqref{eq:h_stationary_sol} is the off-centred form of the centred field $u=\omega x\cot\omega t$ under Proposition~\ref{prop:translation}, and $H_{\rm fan}^{\rm st}$, $U_{\rm fan}^{\rm st}$ are the nontrivial fan profiles, with the run-up asymptotics near the dry edge $X^{\rm lake}_{\rm dry}$.

\begin{figure}[t!]
    \centering
    \includegraphics[width=0.99\linewidth]{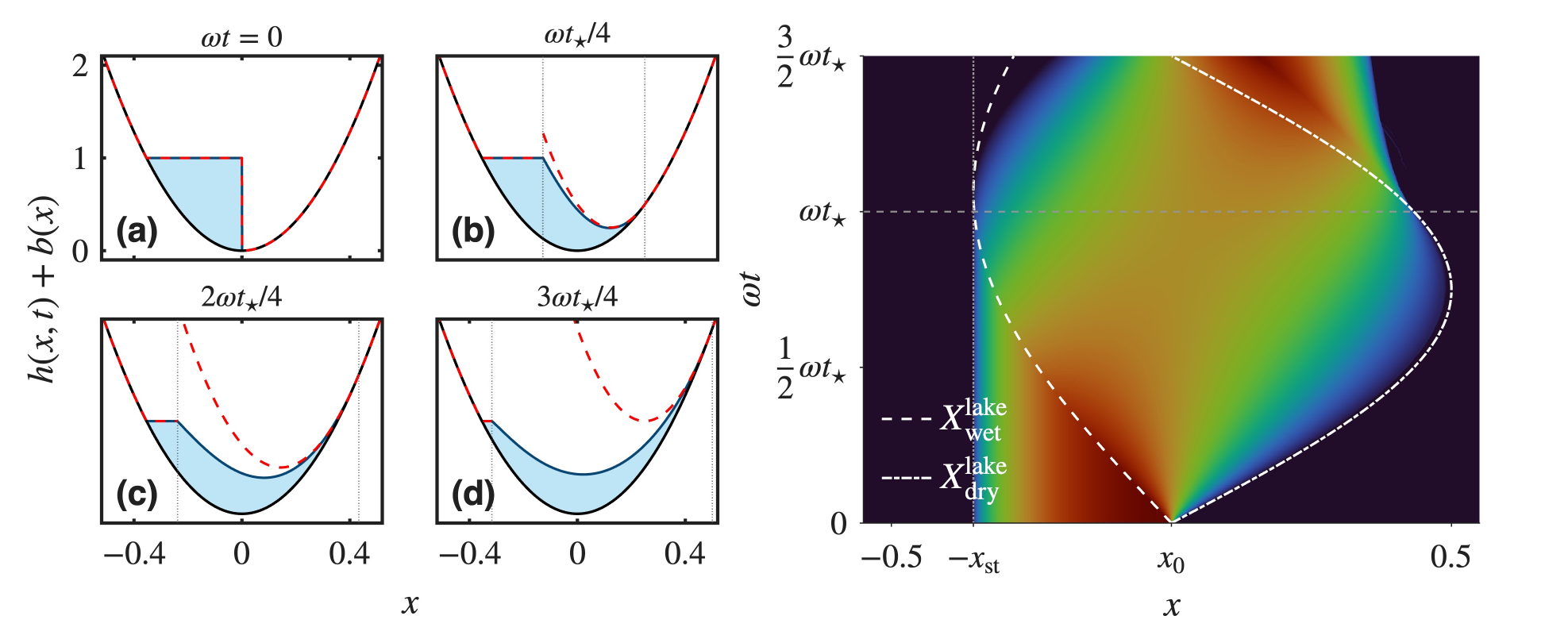}
    \caption{Evolution of the lake-at-rest configuration~\eqref{eq:swe-stationary_data} over a parabolic well for $\delta=0$ and $h_0=1$. The left panels show depth profiles at selected times, while the right panel shows a space-time colour map of $h(x,t)$ in the $(x,\omega t)$-plane. The time $t_\star$ (see Eq.~\eqref{eq:stationary_validity}) is the endpoint of the interval on which the piecewise profile~\eqref{eq:h_stationary_sol} is valid. For this simulated centred case, the maximum vertical elevation attained (run-up) is $2h_0$ above the well bottom.}
    \label{fig:stationary_state_evolution}
\end{figure}

\paragraph{\textbf{Dry edge}}
At the dry front the depth vanishes, so~\eqref{eq:edge_speed} reduces to
$\dot{X}^{\rm lake}_{\rm dry}(t;\omega)=u$, the edge being carried by the off-centred vacuum
velocity field of~\eqref{eq:h_stationary_sol}. The initial location of the edge is  $X^{\rm lake}_{\rm dry}(0;\omega)=x_0=\delta x_{\rm st}$ with initial edge speed fixed by the dry-bed release, which is twice the local wave speed behind the jump,
\begin{equation}\label{eq:dry_edge_launch}
    \dot{X}^{\rm lake}_{\rm dry}(0;\omega)=2\sqrt{h_{\rm st}(x_0)}=2c_0\sqrt{1-\delta^{2}} .
\end{equation}
This gives the dry edge
\begin{equation}\label{eq:stationary_right_edge}
    X^{\rm lake}_{\rm dry}(t;\omega)
    = \delta x_{\rm st} \cos(\omega t)+\sqrt{2}x_{\rm st}\,\sqrt{1-\delta^{2}}\,\sin(\omega t).
\end{equation}
For $\delta=0$ this recovers the centred dry edge $\tfrac{2c_0}{\omega}\sin\omega t$, and the launch speed vanishes at both $|\delta|\to1$. As $\delta\to1$ the jump sits at $x_{\rm st}$, where the lake-at-rest depth vanishes, so the wet region is the full equilibrium profile on $[-x_{\rm st},x_{\rm st}]$ and nothing moves. As $\delta\to-1$ the wet region $[-x_{\rm st},x_0]$ shrinks to a point, and the jump disappears.

\paragraph{\textbf{Wave run--up}}
The dry edge is the moving shoreline, where $h=0$ so the free-surface elevation $\eta=h+b$ is set entirely by the bathymetry $b(X^{\rm lake}_{\rm dry}(t;\omega))$, which gives the vertical elevation of the shoreline with respect to the bottom of the well. The maximum run-up height above the still-water level $h_0$ is obtained as
\begin{equation}\label{eq:stationary_runup}
    \max_t\,\bigl[\,b\bigl(X^{\rm lake}_{\rm dry}(t;\omega)\bigr)-h_0\,\bigr] = h_0\bigl(1-\delta^2\bigr) >0 .
\end{equation}
The positivity shows that, for $-1<\delta<1$, the shoreline always rises above the still-water level. For the centred case $\delta=0$ this is a run-up height $h_0$, i.e. an elevation $2h_0$ above the well bottom.

\paragraph{\textbf{Wet edge}}
The lake edge is determined by matching to the stationary region. Since the fluid is at rest there, the slow characteristic speed at the lake edge equals $\lambda_- = -\sqrt{h_{\rm st}(X^{\rm lake}_{\rm wet}(t;\omega))}$, giving
\begin{equation}\label{eq:left_edge_ode}
    \dot{X}^{\rm lake}_{\rm wet}(t;\omega) = -\sqrt{h_0 - \tfrac{1}{2}\omega^2 \bigl(X^{\rm lake}_{\rm wet}(t;\omega)\bigr)^2},
    \qquad X^{\rm lake}_{\rm wet}(0;\omega) = x_0,
\end{equation}
whose solution is
\begin{equation}\label{eq:stationary_left_edge}
    X^{\rm lake}_{\rm wet}(t;\omega) = \delta x_{\rm st} \cos\!\left(\frac{\omega t}{\sqrt{2}}\right) - x_{\rm st}\sqrt{1-\delta^2}\,\sin\!\left(\frac{\omega t}{\sqrt{2}}\right).
\end{equation}

The two edges have independent validity intervals. As the fan expands, its lake edge $X^{\rm lake}_{\rm wet}(t;\omega)$ advances leftward until the undisturbed stationary region is exhausted at $t=t_{\rm st}$, when $X^{\rm lake}_{\rm wet}$ reaches the left shore $x=-x_{\rm st}$ at which the lake-at-rest depth vanishes, that is, $X^{\rm lake}_{\rm wet}(t_{\rm st};\omega)=-x_{\rm st}$. Independently, the fan profile $H_{\rm fan}$, and with it the run-up solution, remains regular only up to its singularity at $t=t_{\rm run-up}$. The piecewise representation~\eqref{eq:h_stationary_sol} is therefore valid on
\begin{equation}\label{eq:stationary_validity}
    0 < t < t_\star,
    \qquad
    t_\star = \min(t_{\rm st}, t_{\rm run-up}),
\end{equation}
where
\begin{equation}\label{eq:stationary_times}
    \omega t_{\rm st} = \frac{\pi}{\sqrt{2}}\!\left(1 + \frac{2}{\pi}\arcsin\delta\right),
    \qquad
    \omega t_{\rm run-up} = \frac{2\pi}{3}.
\end{equation}
Equating the two times gives the threshold
\begin{equation}\label{eq:stationary_delta_star}
    \delta_* = \sin\!\left(\frac{2\sqrt{2}-3}{6}\,\pi\right) \approx -0.0897.
\end{equation}
For $\delta < \delta_*$, the stationary region is exhausted first and $t_\star=t_{\rm st}$. Otherwise, for $\delta > \delta_*$ the run-up singularity is reached first and $t_\star=t_{\rm run-up}$.

\subsection[Applications to Bose-Einstein Condensates]{Applications to Bose-Einstein Condensates}
\label{sec:NLS}

Bose--Einstein condensates (BECs) of ultracold atoms are described by the
Gross--Pitaevskii (GP) equation, a form of the nonlinear Schr\"odinger (NLS)
equation~\cite{el2016dam}. In many experiments the atoms are tightly confined in a
``cigar''-shaped trap, whose dynamics is described by the one-dimensional GP
equation with a harmonic potential~\cite{sharan2025breaking},
\begin{equation}\label{eq:1d-gp}
    i\hbar\,\frac{\partial \psi}{\partial t}
    = -\frac{\hbar^{2}}{2m}\frac{\partial^{2}\psi}{\partial x^{2}}
    + \frac{1}{2}\, m\,\omega^{2}\, x^{2}\,\psi
    + g\,|\psi|^{2}\psi,
\end{equation}
where $\hbar$ is the reduced Planck constant, $m$ is the atomic mass, $g>0$ is the
repulsive effective one-dimensional interaction strength, and $\tfrac12 m\omega^2 x^2$ is the
axial harmonic trap.

Writing the wavefunction in hyrdodynamic variables
$\psi = \sqrt{n}\,e^{im\phi/\hbar}$, with $n=|\psi|^{2}$ the atomic density and
$u=\partial_x\phi$ the flow velocity, the GP equation~\eqref{eq:1d-gp} can be recast as
a system of two conservation laws for $n$ and $u$. Introducing the scaled variables
\begin{equation}\label{eq:gp-swe-map}
    h \;\doteq\; \frac{g}{m}\,n,
    \qquad
    c_s \;=\; \sqrt{\frac{g\,n}{m}} \;=\; \sqrt{h},
\end{equation}
where $h$ is the analogue of the shallow-water depth and $c_s$ the local speed of sound, the hydrodynamic system takes a form closely resembling the SWE,
\begin{subequations}\label{eq:gp-as-swe}
\begin{align}
    \frac{\partial h}{\partial t} + \frac{\partial(hu)}{\partial x} &= 0,
    \label{eq:gp-as-swe-mass}\\[2mm]
    \frac{\partial u}{\partial t} + u\,\frac{\partial u}{\partial x} + \frac{\partial h}{\partial x}
    &= -\,\omega^{2}x
    + \frac{\hbar^{2}}{2m^{2}}\,
        \frac{\partial}{\partial x}\!\left(\frac{1}{\sqrt{h}}\,\frac{\partial^{2}\sqrt{h}}{\partial x^{2}}\right).
    \label{eq:gp-as-swe-mom}
\end{align}
\end{subequations}
The continuity equation~\eqref{eq:gp-as-swe-mass} expresses mass conservation, and the atomic density is recovered as $n=(m/g)\,h$. In the momentum balance~\eqref{eq:gp-as-swe-mom}, the term $\partial_x h$ is the mean-field pressure, the term $\propto\omega^{2}$ is the harmonic-trap forcing, and the dispersive term $\propto\hbar^{2}$ is the \emph{quantum pressure}.
The harmonic trap frequency $\omega$ plays the role of the bathymetric curvature. 

Apart from the quantum-pressure term,~\eqref{eq:gp-as-swe} is precisely the forced SWE with parabolic bathymetry~\eqref{eq:swe-parabolic} studied in this work. 
Whenever the local sound speed is spatially linear, such that
\begin{equation}\label{eq:gp-affine}
    \sqrt{h(x,t)} = a(t)\,x + b(t),
\end{equation}
one has $\partial_x^{2}\sqrt{h}=0$, so the quantum-pressure term in~\eqref{eq:gp-as-swe} vanishes \emph{identically}. The exact solutions presented in this work --- the run-up, spatially uniform and vacuum solutions (\S\ref{sec:main_swe}) --- are all of this form, and are therefore also \emph{exact} solutions of the GP / NLS equation~\eqref{eq:1d-gp}.

A dry dam-break experiment was recently realised in a BEC~\cite{sharan2025breaking}. In the experiment, the condensate is first prepared by confining atoms in one half of the harmonic trap and allowing them to relax to a stationary state, with the density balancing the trapping potential and the flow at rest. This is the direct analogue of the lake-at-rest configuration of \S\ref{sec:swe_stationary_profile}, in which water fills one side of the parabolic bathymetry in hydrostatic equilibrium. Once released, the run-up solution~\eqref{eq:main-run_up} and the dry edge~\eqref{eq:dry-edge} characterise the condensate near the vacuum point.

Since the run-up solution is only a local description near the vacuum edge, a global description across the whole rarefaction region was constructed in~\cite{sharan2025breaking} as a \emph{matched solution}. As the run-up profiles are spatially linear in $\chi$, a quadratic spatial profile with time-dependent coefficients was considered. This ansatz expressed in the wavefront coordinate $\chi \doteq X_{\rm dry}(t;\omega) - x$, takes the form
\begin{equation}\label{eq:gp-matched}
    c_s \approx c_0(t)+c_1(t)\,\chi+c_2(t)\,\chi^{2},
    \qquad
    u \approx u_0(t)+u_1(t)\,\chi+u_2(t)\,\chi^{2},
\end{equation}
where the leading two coefficients $c_0$, $c_1$ and $u_0$, $u_1$ are obtained directly from the run-up solution~\eqref{eq:main-run_up} by expressing it in terms of $\chi$ and comparing coefficients, while $c_2$ and $u_2$ are determined uniquely by requiring the matched profiles connect continuously with the stationary profile at the wet edge. Experimental measurements of the BEC density and flow velocity, together with direct numerical simulations of the (3+1)D GP equation, were found to be in excellent agreement with this theory.

\subsection[Parabolic hill bathymetry]{Parabolic hill bathymetry}\label{sec:swe_repulsive}

In contrast to the parabolic well bathymetry considered in \S\ref{sec:swe}, we now consider the parabolic hill
\begin{equation}\label{eq:parabolic_hill_bathymetry}
    b(x) = -\tfrac{1}{2}\omega^2 x^2.
\end{equation}
The curvature has the opposite sign, so the bathymetric forcing expels fluid outward rather than trapping it inside the parabolic well. 
The branch formulas used below are obtained from the parabolic-well formulas by the continuation $\omega\mapsto i\omega$, which replaces the trigonometric factors by their hyperbolic counterparts. We record only the depth profiles; the velocity profiles are obtained in the same way from the corresponding attractive formulas. The exterior uniform-depth and vacuum branches have depths
\begin{equation}\label{eq:swe_uniform_phys_repulsive}
    h(x,t) = h_0\,\mathrm{sech}(\omega t)
    \qquad \text{(uniform depth)},
\end{equation}
\begin{equation}\label{eq:swe_vacuum_phys_repulsive}
    h(x,t) = 0
    \qquad \text{(vacuum)}.
\end{equation}
We use tildes for the repulsive run-up and fan profiles to distinguish them from the attractive profiles in Section~\S\ref{sec:swe}. The repulsive run-up depth is
\begin{equation}\label{eq:swe_slow_phys_repulsive}
    \thrunup(x,t;\omega) = \frac{\omega^{2}}{4}\,\mathrm{csch}^{2}\!\left(\frac{3\omega t}{2}\right)
        \!\left(\frac{2c_0}{\omega}\sinh(\omega t)-x\right)^{2}.
\end{equation}
The formulas above are regular for all $t>0$. 

\subsubsection{Lake-at-rest configuration}\label{sec:swe_repulsive_riemann}

This is the repulsive counterpart of the lake-at-rest dam-break considered in \S\ref{sec:swe_stationary_profile}. The corresponding lake-at-rest profile is
\begin{equation}\label{eq:h_stationary_repulsive}
    \tilde{h}_{\rm st}(x) = h_0 + \tfrac{1}{2}\omega^2 x^2.
\end{equation}
Unlike the attractive profile~\eqref{eq:h_stationary_profile}, $\tilde{h}_{\rm st}$ is positive for all $x$. For a centred dam at the hilltop, with vacuum to the right, the depth profile has the form
\begin{equation}\label{eq:swe_repulsive_riemann_sol}
    h(x,t) =
    \begin{cases}
        \tilde{h}_{\rm st}(x),
            & x < \tilde{X}^{\rm lake}_{\rm wet}(t;\omega), \\[1mm]
        \tilde{H}_{\rm fan}(x,t;\omega),
            & \tilde{X}^{\rm lake}_{\rm wet}(t;\omega) < x < \tilde{X}^{\rm lake}_{\rm dry}(t;\omega), \\[1mm]
        0,
            & x > \tilde{X}^{\rm lake}_{\rm dry}(t;\omega),
    \end{cases}
\end{equation}
where the fan depth $\tilde{H}_{\rm fan}$ has the run-up behaviour~\eqref{eq:swe_slow_phys_repulsive} near the vacuum edge. The corresponding dry and wet edges are
\begin{align}
    \tilde{X}^{\rm lake}_{\rm dry}(t;\omega)
    &=
    \frac{2\sqrt{h_0}}{\omega}\sinh(\omega t),
    \label{eq:repulsive_right_edge}
    \\[2mm]
    \tilde{X}^{\rm lake}_{\rm wet}(t;\omega)
    &=
    -\frac{\sqrt{2h_0}}{\omega}
    \sinh\!\left(\frac{\omega t}{\sqrt{2}}\right).
    \label{eq:swe_repulsive_left_edge_st}
\end{align}
In the parabolic-hill case, the run-up branch has no finite-time singularity and the stationary profile remains positive, so~\eqref{eq:swe_repulsive_riemann_sol} is valid for all $t>0$. The evolution is shown in Figure~\ref{fig:repulsive_lake_at_rest}.

\begin{figure}[!h]
    \centering
    \includegraphics[width=0.99\linewidth]{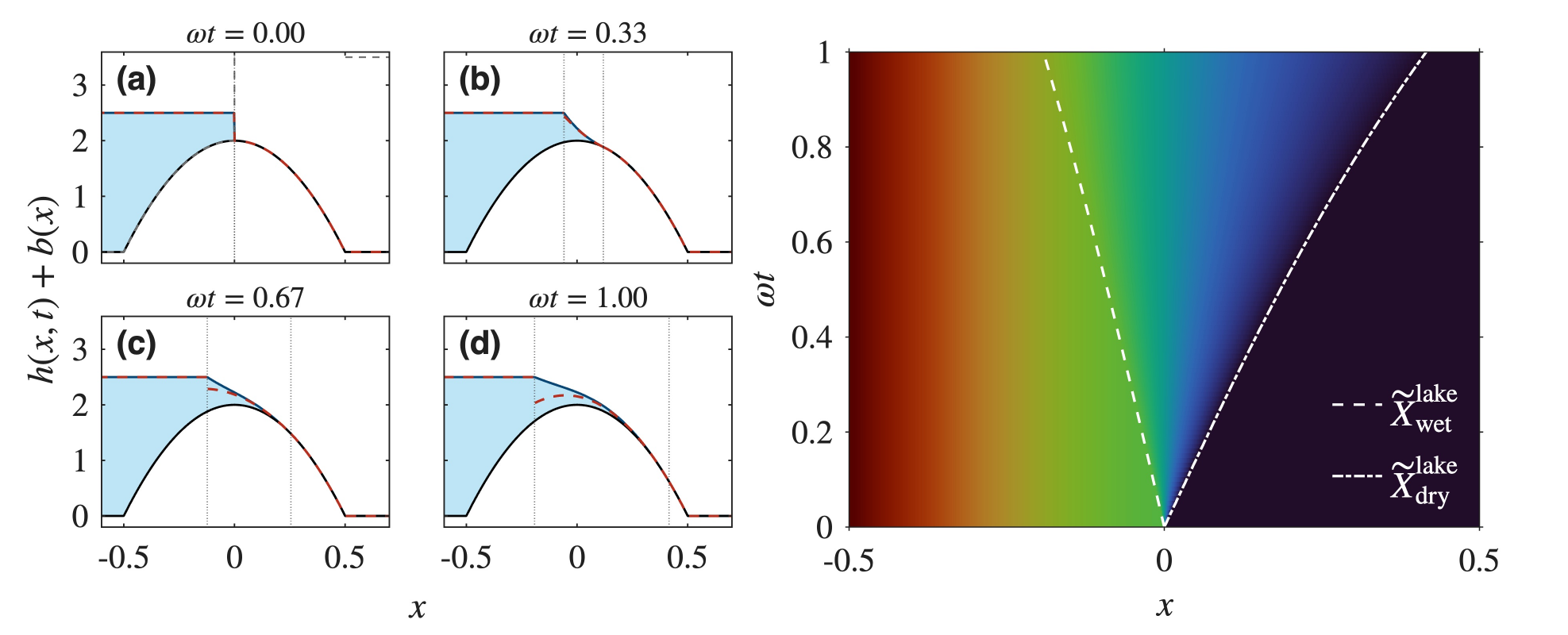}
    \caption{Evolution of the lake-at-rest configuration over a repulsive parabolic hill, $b(x) = -\frac{1}{2}\omega^2 x^2$. The left panels show depth profiles at selected times, while the right panel shows a space-time colour map of the depth in the $(x,\omega t)$-plane. The analytical solution (red dashed) and edge locations (dotted) are overlaid. The dry edge accelerates exponentially outward and the undisturbed lake-at-rest region recedes to the left.}
    \label{fig:repulsive_lake_at_rest}
\end{figure}

\newpage
\subsubsection{Box-type configuration}\label{sec:swe_repulsive_hat}

We consider the box-configuration initial data~\eqref{eq:swe-hat_data} over the repulsive hill, the counterpart of \S\ref{sec:swe_hat}. Two rarefaction fans emanate from the left and right jumps, each described locally by the translated repulsive run-up branch~\eqref{eq:swe_slow_phys_repulsive}. The depth profile is
\begin{equation}\label{eq:h_hat_sol_repulsive}
    h(x,t) =
    \begin{cases}
        0,
            & x < \tilde{X}_{\rm dry}^{L}(t;\omega), \\[1mm]
        \tilde{H}_{\rm fan}^{L}(x,t),
            & \tilde{X}_{\rm dry}^{L}(t;\omega) < x < \tilde{X}_{\rm wet}^{L}(t;\omega), \\[1mm]
        h_0\,\mathrm{sech}(\omega t),
            & \tilde{X}_{\rm wet}^{L}(t;\omega) < x < \tilde{X}_{\rm wet}^{R}(t;\omega), \\[1mm]
        \tilde{H}_{\rm fan}^{R}(x,t),
            & \tilde{X}_{\rm wet}^{R}(t;\omega) < x < \tilde{X}_{\rm dry}^{R}(t;\omega), \\[1mm]
        0,
            & x > \tilde{X}_{\rm dry}^{R}(t;\omega),
    \end{cases}
\end{equation}
The edge laws follow from the continued dry-edge and uniform-depth branches. Define
\begin{equation}\label{eq:F_sech_repulsive}
    \tilde{\cF}(z) = \int_{0}^{z}\mathrm{sech}^{3/2}\theta\,\mathrm{d}\theta.
\end{equation}
The four edges are
\begin{equation}\label{eq:edges_hat_repulsive}
\begin{aligned}
    \tilde{X}_{\rm dry}^{L}(t;\omega)  &= -\frac{2c_0}{\omega}\sinh(\omega t) + x_{0}^{L}\cosh(\omega t),
    & \tilde{X}_{\rm dry}^{R}(t;\omega) &= \phantom{-}\frac{2c_0}{\omega}\sinh(\omega t) + x_{0}^{R}\cosh(\omega t),\\[2mm]
    \tilde{X}_{\rm wet}^{L}(t;\omega)  &=\Bigl[\frac{c_0}{\omega}\,\tilde{\cF}(\omega t) + x_{0}^{L}\Bigr]\cosh(\omega t),
    & \tilde{X}_{\rm wet}^{R}(t;\omega) &= \Bigl[-\frac{c_0}{\omega}\,\tilde{\cF}(\omega t) + x_{0}^{R}\Bigr]\cosh(\omega t).
\end{aligned}
\end{equation}
Between the fans, the constant-depth branch decays as $h_0\,\mathrm{sech}(\omega t)$. Since $\mathrm{sech}^{3/2}$ is integrable on $[0,\infty)$, $\tilde{\cF}$ converges, so collision of the wet edges is not forced as in the attractive case; it depends on the initial plateau width. In the configuration shown in Figure~\ref{fig:repulsive_hat}, the fans spread apart and the plateau decays monotonically.

\begin{figure}[!h]
    \centering
    \includegraphics[width=0.99\linewidth]{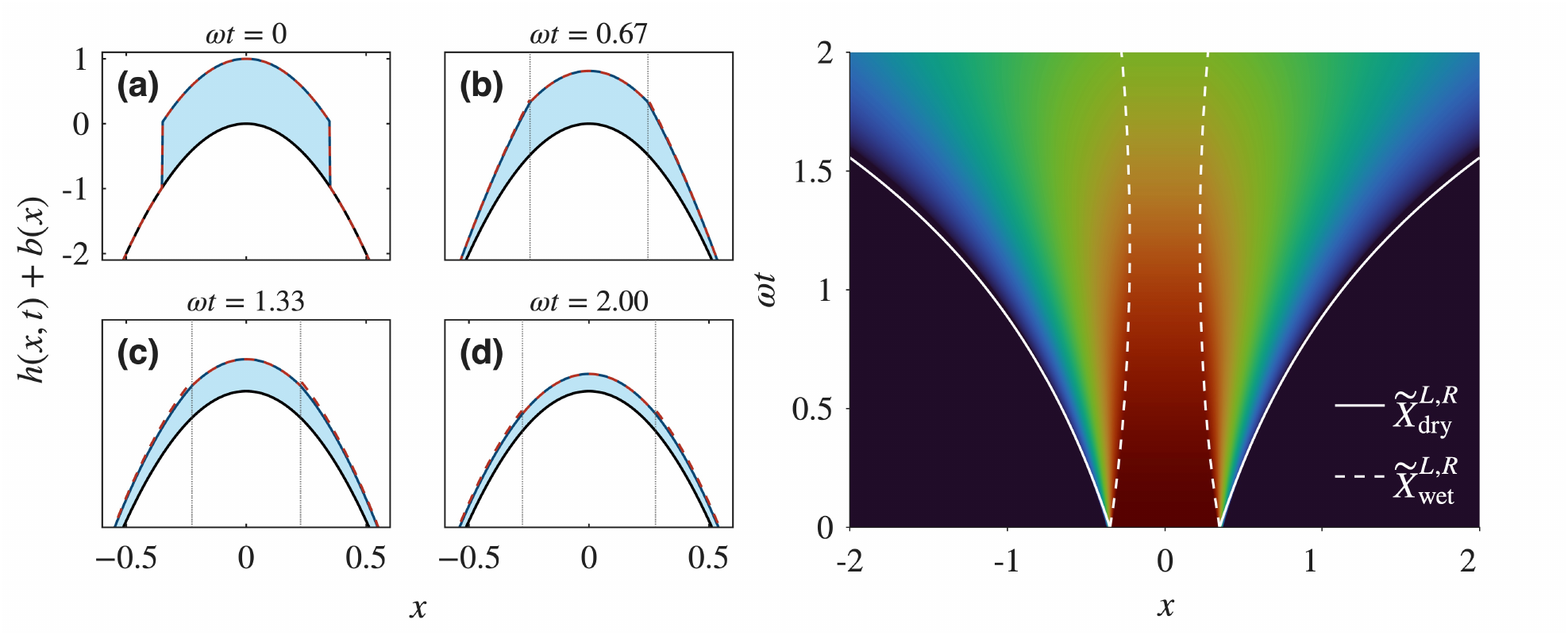}
    \caption{Evolution of the box-configuration dam-break over a repulsive parabolic hill. The left panels show depth profiles at selected times, while the right panel shows a space-time colour map of the depth in the $(x,\omega t)$-plane. The analytical run-up solution (red dashed) and edge locations (dotted) are overlaid. The two fans spread apart and the plateau decays monotonically.}
    \label{fig:repulsive_hat}
\end{figure}

%% file: sections/6_conclusion.tex
\section{Concluding remarks and future directions}\label{sec:conclusion}

We developed a perturbative framework that extends the classical rarefaction wave solutions of Burgers’ equation and the 1D SWE to parabolic forcing. For Burgers’ equation with Riemann initial data, the framework yields an exact global solution. For the SWE, the same approach yields several exact branches of the Riemann solution. The spatially uniform branches remain globally valid. The run-up branch, which generalises the unforced rarefaction fan to parabolic bathymetry, is likewise exact but describes the rarefaction dynamics asymptotically near the dry edge. For one-sided Riemann data, these rarefaction solutions remain valid up to the first singularity time. Together, these results provide exact solutions for nonlinear wave propagation over parabolic bathymetry and establish a framework for extending classical Riemann problems to spatially confined media.

We remark that this work has a close connection to the work of Camassa \emph{et al.}~\cite{camassa2022evolution}, where solutions were sought using a spatially quadratic depth profile together with a linear velocity profile. In terms of the Riemann variables, this choice is equivalent to the affine form~\eqref{eq:affine_form}, which is rigorously obtained here as a consequence of analyticity in \(\omega\). Accordingly, the coefficient ODE system~\eqref{eq:ABCD_ODE_swe} obtained here is equivalent to the one studied there. In particular, the closed \(B_{\pm}\) subsystem~\eqref{eq:ABCD_ODE_swe_Bpm} maps onto the nonlinear \((\alpha,\gamma)\) system studied there under \(B_{\pm}/t=\alpha\pm2\sqrt{\gamma}\), and the sign of \(\gamma\) selects between two distinct types of dynamics. For \(\gamma>0\), \(B_{\pm}\) are real-valued, corresponding to rarefaction-type solutions constructed in the present work, whose curvature grows without bound up to the singularity time~\eqref{eq:singularity_time_swer}. Camassa \emph{et al.} refer to this as the blow-up solution. For \(\gamma<0\), the variables \(B_{+}\) and \(B_{-}\) become complex conjugates, so the system reduces to a single complex-valued ODE. These solutions describe the sloshing of a fluid of finite mass, bounded by two vacuum contact points in the parabolic basin, with no singularity at any time. Camassa \emph{et al.} refer to this as the sloshing droplet solution. 

Beyond shallow-water flows, the exact solutions presented here may extend to several other physical settings in which shallow-water equations arise as the leading-order description of weakly dispersive, smooth flow. In the weakly dispersive regime, the hydrodynamic formulation of the Gross–Pitaevskii/nonlinear Schr\"odinger equation reduces to a shallow-water-type system. In the Bose–Einstein condensate setting, the parabolic bathymetry then plays the role of the harmonic trapping potential used to confine the condensate. The rarefaction-wave solutions derived here have found quantitative experimental validation in trapped Bose–Einstein condensates and predict several key dynamical features of the condensate flow induced by harmonic confinement~\cite{sharan2025breaking}. The same reduction also underlies nonlinear optics, where the governing model is again the nonlinear Schr\"odinger equation and a quadratically varying refractive-index profile acts as an effective harmonic confining potential. Dam-break flows in homogeneous optical media have been investigated experimentally in photon fluids~\cite{xu2017dispersive}. The solutions we provide may serve as an analytical benchmark for extending such studies to inhomogeneous optical media.

We conclude by outlining a mathematical direction suggested by the present work. The exact run-up branch describes the asymptotic neighbourhood of the dry edge, but does not provide a global description of the confined rarefaction fan. Viewed from the perspective of the local wavefront expansion near the dry edge, the unconfined flat-bathymetry rarefaction wave appears as the exceptional case in which the expansion truncates exactly. All higher-order terms vanish, and the local solution extends to the entire rarefaction region, connecting continuously to the exterior state. Confinement fundamentally changes this picture, since the higher-order terms do not vanish and contribute throughout the fan. Appendix~\ref{apdx_sec:wavefront_higher_order} shows that these corrections are non-analytic in \(\omega\) and introduce additional degrees of freedom into the local solution. Although, as discussed in \S\ref{sec:NLS}, an approximate quadratic global rarefaction profile may be constructed by continuous matching of the local solution branches, the exact mathematical structure of the confined fan remains unresolved. 
In particular, whether and how the degrees of freedom introduced by the non-analytic corrections may be used to match to an exterior state, thereby completing the global construction of the confined rarefaction fan, remains an open question.

%% file: sections/appendix.tex
\newpage
\begin{appendix}

\section[Appendix -- Perturbative calculation for Burgers' equation]{Perturbative calculation for Burgers' equation}\label{apdx:burgers}

Here we provide the details of the perturbative calculation for the forced Burgers' equation~\eqref{eq:burgers-parabolic}. Substituting the perturbative series~\eqref{eq:burgers-series-expansion} and collecting $O(\omega^{2n})$ terms yields, in each region, a hierarchy of linear PDEs for the corrections $u_n$ with zero initial data.

\subsection{Interior solution}\label{apdx:burgers-interior}

In the interior region the leading-order solution is $u_{0}(x,t) = x/t$. The correction hierarchy is
\begin{subequations}\label{eq:burgers-apdx-interior-hierarchy}
\begin{align}
    \mathcal{L}^{\rm (int)}u_{1} &= -x, \label{eq:burgers-apdx-interior-hierarchy-n1} \\[2mm]
    \mathcal{L}^{\rm (int)}u_{n} &= -\sum_{j=1}^{n-1}u_{j}\,\frac{\partial u_{n-j}}{\partial x}, \qquad n\geq 2, \label{eq:burgers-apdx-interior-hierarchy-nge2}
\end{align}
\end{subequations}
where the transport operator is
\begin{equation}\label{eq:burgers-apdx-interior-operator}
    \mathcal{L}^{\rm (int)} = \partial_t + \tfrac{x}{t}\,\partial_x + \tfrac{1}{t}.
\end{equation}
For $n=1$, the equation~\eqref{eq:burgers-apdx-interior-hierarchy-n1} can be solved by the method of characteristics, giving
\begin{equation}
    u_1(x,t) = -\tfrac{1}{3}\, x\, t = \bigl(F_1\,\tfrac{x}{t}\bigr)\,t^{2}, \qquad F_1 = -\tfrac{1}{3}.
\end{equation}
The method of characteristics can still be used at every order $n \geq 2$, but the calculations become increasingly tedious. Remarkably, the solution form obtained at $n=1$ generalises to order $n$ by multiplication by $t^{2n}$, giving the ansatz
\begin{equation}\label{eq:burgers-apdx-interior-ansatz}
    u_n(x,t) = \bigl(F_n\,\tfrac{x}{t}\bigr)\,t^{2n}.
\end{equation}
If $u_j = \bigl(F_j\,\tfrac{x}{t}\bigr)\,t^{2j}$ for all $j < n$, then the right-hand side of~\eqref{eq:burgers-apdx-interior-hierarchy-nge2} is
\begin{equation}
    -\sum_{j=1}^{n-1} u_j\, \frac{\partial u_{n-j}}{\partial x} = -\tfrac{x}{t}\biggl(\sum_{j=1}^{n-1} F_j\, F_{n-j}\biggr) t^{2n-1},
\end{equation}
while the left-hand side of~\eqref{eq:burgers-apdx-interior-hierarchy-nge2} is
\begin{equation}
    \mathcal{L}^{\rm (int)}\bigl[\bigl(F_n\,\tfrac{x}{t}\bigr)\,t^{2n}\bigr] = \tfrac{x}{t}\,(2n+1)\, F_n\, t^{2n-1}.
\end{equation}
Equating the two, the common factor $\tfrac{x}{t}t^{2n-1}$ cancels, allowing matching the remaining coefficients. This yields the recursion
\begin{equation}\label{eq:apdx_recursion_burgers_Fn_interior}
    F_n = -\frac{1}{2n+1} \sum_{j=1}^{n-1}F_j\,F_{n-j}, \qquad n\geq 2,
\end{equation}
which determines all coefficients.

Substituting the ansatz~\eqref{eq:burgers-apdx-interior-ansatz} back into the perturbative series~\eqref{eq:burgers-series-expansion}, the factors $t^{2n}$ and $\omega^{2n}$ combine into $(\omega t)^{2n}$, giving
\begin{equation}\label{eq:burgers-interior-resummed}
    u(x,t;\omega) = \frac{x}{t}\sum_{n=0}^{\infty} F_n\,(\omega t)^{2n} \doteq \frac{x}{t}\,F(\omega t).
\end{equation}
Setting $z = (\omega t)^{2}$, the function
\begin{equation}\label{eq:burgers-interior-generating-function}
    F(z) = \sum_{n=0}^{\infty} F_n\, z^{n}
\end{equation}
serves as the generating function for the coefficients $F_n$. 

Rearranging~\eqref{eq:apdx_recursion_burgers_Fn_interior}, multiplying by $z^{n}$, and summing from $n = 2$ to $\infty$ gives
\begin{equation}\label{eq:sum_recursion_burger}
    \sum_{n=2}^{\infty}(2n+1)\,F_n\,z^{n} + \sum_{n=2}^{\infty}\biggl(\sum_{j=1}^{n-1}F_j\,F_{n-j}\biggr)z^{n} = 0.
\end{equation}
The first sum is
\begin{equation}\label{eq:sum_ref_burgers_1}
    \sum_{n=2}^{\infty}(2n+1)\,F_n\,z^{n} = 2z\,F'(z) + F(z) - F_0 - 3F_1\,z,
\end{equation}
and the second, may be expressed as the Cauchy product of $F(z)-F_0$ with itself, giving
\begin{equation}\label{eq:sum_ref_burgers_2}
    \sum_{n=2}^{\infty}\biggl(\sum_{j=1}^{n-1}F_j\,F_{n-j}\biggr)z^{n} = \bigl(F(z)-F_0\bigr)^{2} = F^{2} - 2F_0\,F + F_0^{2}.
\end{equation}
Substituting~\eqref{eq:sum_ref_burgers_1} and~\eqref{eq:sum_ref_burgers_2} into~\eqref{eq:sum_recursion_burger} and using $F_0 = 1$, $F_1 = -\tfrac{1}{3}$ gives
\begin{equation}\label{eq:burgers-interior-ode-z}
    2z\,\frac{dF}{dz} - F + F^{2} = -z.
\end{equation}
Substituting $z = (\omega t)^{2}$, so that $t\tfrac{d}{dt} = 2z\,\frac{d}{dz}$, recovers the ODE~\eqref{eq:burgers-interior-ode-F} of Proposition~\ref{prop:burgers-interior}.

\subsection{Exterior solution}\label{apdx:burgers-exterior}

In each exterior region the leading-order solution is the constant state $u_{0}(x,t)=c$, where $c=u_{\pm}$. The correction hierarchy is
\begin{subequations}\label{eq:burgers-apdx-exterior-hierarchy}
\begin{align}
    \mathcal{L}^{\rm (ext)}\, u_{1} &= -x, \label{eq:burgers-apdx-exterior-hierarchy-n1} \\[2mm]
    \mathcal{L}^{\rm (ext)}\, u_{n} &= -\sum_{j=1}^{n-1}u_{j}\,\frac{\partial u_{n-j}}{\partial x},
    \qquad n\geq 2, \label{eq:burgers-apdx-exterior-hierarchy-nge2}
\end{align}
\end{subequations}
where the transport operator is
\begin{equation}\label{eq:burgers-apdx-exterior-operator}
    \mathcal{L}^{\rm (ext)} = \partial_t + c\,\partial_x.
\end{equation}
For $n=1$, the equation~\eqref{eq:burgers-apdx-exterior-hierarchy-n1} can be solved by the method of characteristics with $u_1(x,0) = 0$, giving
\begin{equation}
    u_1(x,t) = \bigl(-\tfrac{x}{t} + \tfrac{1}{2}\,c\bigr)\,t^{2},
\end{equation}
which has the form $u_1 = (F_1\, \tfrac{x}{t} + c\, G_1)\, t^{2}$ with $F_1 = -1$ and $G_1 = \tfrac{1}{2}$. As in the interior region, the $n=1$ form generalises to order $n$ by multiplication by $t^{2n}$, giving the ansatz
\begin{equation}\label{eq:burgers-apdx-exterior-ansatz}
    u_n(x,t) = (F_n\, \tfrac{x}{t} + c\, G_n)\, t^{2n}.
\end{equation}
If $u_j = (F_j\, \tfrac{x}{t} + c\, G_j)\, t^{2j}$ for all $j < n$, then the right-hand side of~\eqref{eq:burgers-apdx-exterior-hierarchy-nge2} is
\begin{equation}
    -\sum_{j=1}^{n-1} u_j\, \frac{\partial u_{n-j}}{\partial x} = -\biggl[\tfrac{x}{t}\sum_{j=1}^{n-1} F_j\, F_{n-j} + c\sum_{j=1}^{n-1} F_{n-j}\, G_j\biggr]\, t^{2n-1},
\end{equation}
while the left-hand side of~\eqref{eq:burgers-apdx-exterior-hierarchy-nge2} is
\begin{equation}
    \mathcal{L}^{\rm (ext)}\bigl[(F_n\, \tfrac{x}{t} + c\, G_n)\, t^{2n}\bigr] = \biggl[(2n-1)\, F_n\, \tfrac{x}{t} + (2n\, G_n + F_n)\, c\biggr]\, t^{2n-1}.
\end{equation}
Equating the two, the common factors $\tfrac{x}{t}t^{2n-1}$ and $ct^{2n-1}$ cancel to yield the recursion
\begin{equation}
    F_n = -\frac{1}{2n-1} \sum_{j=1}^{n-1}F_j\,F_{n-j}, \qquad 2n\,G_n = - F_n-\sum_{j=1}^{n}F_{j}\,G_{n-j}, \qquad n\geq 2.
\end{equation}

The recursion for \(G_n\) may be simplified further as follows
\begin{equation}
    2n\, G_n = -F_n - \sum_{j=1}^{n-1}F_{n-j}\,G_j = -F_n\,G_0 - \sum_{j=1}^{n-1}F_{n-j}\,G_j = -\sum_{j=0}^{n-1}F_{n-j}\,G_j = -\sum_{k=1}^{n}F_{k}\,G_{n-k},
\end{equation}
where we used $G_0 = 1$ and reindexed $k = n - j$. The recursions are therefore
\begin{equation}\label{eq:burgers-apdx-exterior-recursion}
    F_n = -\frac{1}{2n-1} \sum_{j=1}^{n-1}F_j\,F_{n-j}, \qquad G_n = -\frac{1}{2n}\sum_{j=1}^{n}F_{j}\,G_{n-j}, \qquad n\geq 2.
\end{equation}

Substituting the ansatz~\eqref{eq:burgers-apdx-exterior-ansatz} back into the series~\eqref{eq:burgers-series-expansion}, the factors $t^{2n}$ and $\omega^{2n}$ combine to give
\begin{align}\label{eq:burgers-exterior-resummed}
    u(x,t;\omega) = \sum_{n=0}^{\infty}\bigl(F_n\,\tfrac{x}{t} + c\,G_n\bigr)(\omega t)^{2n} \doteq \frac{x}{t}\,F(\omega t) + c\,G(\omega t).
\end{align}
Setting $z = (\omega t)^{2}$, the functions
\begin{equation}\label{eq:burgers-exterior-generating-functions}
    F(z) = \sum_{n=0}^{\infty} F_n\, z^{n}, \qquad G(z) = \sum_{n=0}^{\infty} G_n\, z^{n}
\end{equation}
serve as generating functions for the coefficients $F_n$ and $G_n$. Following the same procedure as in~\S\ref{apdx:burgers-interior} of rearranging the recursion~\eqref{eq:burgers-apdx-exterior-recursion} and summing over \(n=2\) to \(\infty\), and expressing in terms of the generating functions, gives
\begin{subequations}\label{eq:burgers-exterior-generating-sums}
\begin{align}
    \sum_{n=2}^{\infty}(2n-1)\,F_n\,z^{n}
    + \sum_{n=2}^{\infty} \sum_{j=1}^{n-1} F_j F_{n-j}z^{n}
    &=
    \big(2z\,F'(z) - F - F_1\,z + F_0\big)
    + \big(F - F_0\big)^2,
    \label{eq:burgers-exterior-generating-sums-F}
    \\[2mm]
    \sum_{n=2}^{\infty}2n \, G_n z^n
    + \sum_{n=2}^{\infty} \sum_{j=1}^{n} F_j G_{n-j} z^n
    &=
    \big(2z G'(z) - 2G_1 z\big)
    + \big( F G - F_0 G - F_1 G_0 z \big).
    \label{eq:burgers-exterior-generating-sums-G}
\end{align}
\end{subequations}

Using $F_0 = 0$, $G_0 = 1$ and $F_1 = -1$, $G_1 = 1/2$ gives
\begin{subequations}\label{eq:burgers-exterior-ode-z}
\begin{align}
    2z\,\frac{dF}{dz} - F + F^{2} &= -z,
    \label{eq:burgers-exterior-ode-z-F}\\[2mm]
    2z\,\frac{dG}{dz} + F\,G &= 0,
    \label{eq:burgers-exterior-ode-z-G}
\end{align}
\end{subequations}
with initial data \(F(0) = 0\) and \(G(0) =1\).
Note that~\eqref{eq:burgers-exterior-ode-z-F} is the same as~\eqref{eq:burgers-interior-ode-z}.

Finally substituting $z = (\omega t)^{2}$ recovers~\eqref{eq:burgers-exterior-G} of Proposition~\ref{prop:burgers-exterior}.

    \section[Appendix -- Perturbative calculation for Shallow water equations]{Perturbative calculation for Shallow Water Equations}\label{apdx:swe}

We derive the rarefaction wave solution of the forced SWE, paralleling the Burgers calculation in Appendix~\ref{apdx:burgers}. We solve the perturbative hierarchy on the slow simple-wave branch (\S\ref{apdx_sec:swe_rw_solution}), resum the recursion into the closed-form ODE system~\eqref{eq:ABCD_ODE_swe} (\S\ref{apdx_sec:swe_ode_derivation}), show that the same ODE governs all four branches of the Riemann solution (\S\ref{apdx_sec:Recursion_Relation_ABCD}), and establish the invariant manifolds that yield the closed-form solutions (\S\ref{apdx_sec:BD_system}).

    \subsection{Perturbative hierarchy on the slow simple wave branch}\label{apdx_sec:swe_rw_solution}

We expand the Riemann variables in even powers of $\omega$,
\begin{equation}\label{eq:swe_apdx_expansion}
    R_{\pm}(x,t;\omega) = \sum_{n=0}^{\infty} r^{\pm}_{n}(x,t)\,\omega^{2n},
\end{equation}
with $r^{\pm}_{0}$ taken from the Stoker solution and $r^{\pm}_{n}(x,0) = 0$ for $n \geq 1$. On the slow simple-wave branch the leading-order terms are
\begin{equation}
    r^{+}_{0} = 2c_0, \qquad r^{-}_{0} = -\tfrac{2}{3}c_0 + \tfrac{4}{3}\,\tfrac{x}{t}.
\end{equation}
Substituting into the diagonal form of the forced SWE~\eqref{eq:diagonal_SWE} and collecting $O(\omega^{2n})$ terms, we obtain the coupled hierarchy. For $n=1$,
\begin{subequations}\label{eq:swe_apdx_hierarchy}
\begin{align}
    \mathcal{L}^{(+)} r^{+}_{1} &= -x, \label{eq:swe_apdx_hierarchy_p_n1} \\[2mm]
    \mathcal{L}^{(-)} r^{-}_{1} + \tfrac{1}{3t}\,r^{+}_{1} &= -x, \label{eq:swe_apdx_hierarchy_m_n1}
\end{align}
\end{subequations}
and for $n \geq 2$,
\begin{subequations}\label{eq:swe_apdx_hierarchy_nge2}
\begin{align}
    \mathcal{L}^{(+)} r^{+}_{n} &= -\sum_{j=1}^{n-1}\lambda^{+}_{j}\,\partial_x r^{+}_{n-j}, \label{eq:swe_apdx_hierarchy_p_nge2} \\[2mm]
    \mathcal{L}^{(-)} r^{-}_{n} + \tfrac{1}{3t}\,r^{+}_{n} &= -\sum_{j=1}^{n-1}\lambda^{-}_{j}\,\partial_x r^{-}_{n-j}, \label{eq:swe_apdx_hierarchy_m_nge2}
\end{align}
\end{subequations}
where the leading-order transport operators and characteristic speeds are
\begin{equation}\label{eq:swe_apdx_transport_ops}
    \mathcal{L}^{(+)} = \partial_t + \lambda^{+}_{0}\,\partial_x, \qquad
    \mathcal{L}^{(-)} = \partial_t + \lambda^{-}_{0}\,\partial_x + \tfrac{1}{t}, \qquad
    \lambda^{+}_{0} = \tfrac{1}{3}\!\left(4c_0 + \tfrac{x}{t}\right), \qquad
    \lambda^{-}_{0} = \tfrac{x}{t},
\end{equation}
and
\begin{equation}\label{eq:swe_apdx_lambda_n}
    \lambda^{\pm}_{n} = \tfrac{1}{4}(3r^{\pm}_{n}+r^{\mp}_{n}), \qquad n \geq 1.
\end{equation}

For $n=1$, the $r^{+}_{1}$ equation in~\eqref{eq:swe_apdx_hierarchy} is independent and can be solved first. Once $r^{+}_{1}$ is known, the $r^{-}_{1}$ equation decouples and can again be solved independently. 
One may use the method of characteristics and obtain that the solution is
\begin{equation}
    r^{+}_{1} = \bigl(\tfrac{1}{2}c_0 - \tfrac{3}{4}\,\tfrac{x}{t}\bigr)\,t^{2}, \qquad r^{-}_{1} = \bigl(-\tfrac{1}{18}c_0 - \tfrac{1}{4}\,\tfrac{x}{t}\bigr)\,t^{2},
\end{equation}
which have the affine form $r^{\pm}_{1} = (2c_0\, a^{\pm}_{1} + b^{\pm}_{1}\,\tfrac{x}{t})\,t^{2}$ where
\begin{equation}
    a^{+}_{1} = \tfrac{1}{4},\quad b^{+}_{1} = -\tfrac{3}{4},\quad a^{-}_{1} = -\tfrac{1}{36},\quad b^{-}_{1} = -\tfrac{1}{4}.
\end{equation}

As in the Burgers case (\S\ref{apdx:burgers-interior}), the $n=1$ solution form generalises to order $n$ by multiplication by $t^{2n}$, giving the affine ansatz
\begin{equation}\label{eq:swe_apdx_affine_ansatz}
    r^{\pm}_{n}(x,t) = \bigl(2c_0\, a^{\pm}_{n} + b^{\pm}_{n}\,\tfrac{x}{t}\bigr)\,t^{2n},
\end{equation}
which can be verified that solves the hierarchy at every order by induction. 
Substituting~\eqref{eq:swe_apdx_affine_ansatz} into~\eqref{eq:swe_apdx_hierarchy_nge2} and equating the coefficients of $c_0$ and $\tfrac{x}{t}$ yields the coupled recursions,
    \begin{subequations}\label{eq:ABCD_recursion_slow_branch}
    \begin{align}
        b^{+}_{n} &= -\frac{1}{2n-\tfrac{2}{3}}\left[\frac{1}{4}\sum_{j=1}^{n-1} \bigl(3b^{+}_{j} + b^{-}_{j}\bigr) b^{+}_{n-j}\right],\label{eq:ABCD_recursion_slow_branch_Bp}\\[2mm]
        a^{+}_{n} &= -\frac{1}{2n}\left[\frac{2}{3}b^{+}_{n} + \frac{1}{4}\sum_{j=1}^{n-1} \bigl(3a^{+}_{j} + a^{-}_{j}\bigr) b^{+}_{n-j}\right],\label{eq:ABCD_recursion_slow_branch_Ap}\\[2mm]
        b^{-}_{n} &= -\frac{1}{2n+1}\left[\frac{1}{3}b^{+}_{n} + \frac{1}{4}\sum_{j=1}^{n-1} \bigl(b^{+}_{j} + 3b^{-}_{j}\bigr) b^{-}_{n-j}\right],\label{eq:ABCD_recursion_slow_branch_Bm}\\[2mm]
        a^{-}_{n} &= -\frac{1}{2n+1}\left[\frac{1}{3}a^{+}_{n} + \frac{1}{4}\sum_{j=1}^{n-1} \bigl(a^{+}_{j} + 3a^{-}_{j}\bigr) b^{-}_{n-j}\right],\label{eq:ABCD_recursion_slow_branch_Am}
    \end{align}
    \end{subequations}
for $n \geq 2$. These determine all coefficients in the perturbative expansion.

Substituting~\eqref{eq:swe_apdx_affine_ansatz} back into the expansion~\eqref{eq:swe_apdx_expansion}, the factors $t^{2n}$ and $\omega^{2n}$ combine into $(\omega t)^{2n}$, giving the solution form
\begin{align}\label{eq:swe_apdx_resummed_form}
    R_{\pm}(x,t;\omega) = \sum_{n=0}^{\infty}\bigl(2c_0\,a^{\pm}_{n} + b^{\pm}_{n}\,\tfrac{x}{t}\bigr)(\omega t)^{2n}
                        \doteq 2c_0\,A_{\pm}(\omega t) + \tfrac{x}{t}\,B_{\pm}(\omega t).
\end{align}
This is the affine solution form presented in~\eqref{eq:affine_form}.

    \subsection{ODE system for the branches of the Riemann solution}\label{apdx_sec:swe_ode_derivation}

To derive the ODE system governing $A_{\pm}$ and $B_{\pm}$, we use the recursion relations~\eqref{eq:ABCD_recursion_slow_branch}. The solution form~\eqref{eq:swe_apdx_resummed_form} is a series in $(\omega t)^{2n}$. As in the Burgers case (\S\ref{apdx:burgers-interior}), we set $z = (\omega t)^{2}$ so that
\begin{equation}\label{eq:swe_generating_functions}
    A_{\pm}(z) = \sum_{n=0}^{\infty} a^{\pm}_{n}\, z^{n}, \qquad B_{\pm}(z) = \sum_{n=0}^{\infty} b^{\pm}_{n}\, z^{n}.
\end{equation}
As in the Burgers case, $A_{\pm}$ and $B_{\pm}$ serve as the generating functions for the coefficients $a^{\pm}_{n}$ and $b^{\pm}_{n}$. Rearranging each recursion in~\eqref{eq:ABCD_recursion_slow_branch}, multiplying by $z^{n}$, and summing from $n = 2$ to $\infty$ gives
\begin{subequations}\label{eq:swe_recursion_summed}
\begin{align}
    \sum_{n=2}^{\infty}2n\,b^{+}_{n}\,z^{n}
    - \frac{2}{3}\sum_{n=2}^{\infty}b^{+}_{n}\,z^{n}
    + \frac{1}{4}\sum_{n=2}^{\infty}\biggl(\sum_{j=1}^{n-1}\bigl(3b^{+}_{j}+b^{-}_{j}\bigr)\,b^{+}_{n-j}\biggr)z^{n} &= 0, \label{eq:swe_recursion_summed_Bp}\\[3mm]
    \sum_{n=2}^{\infty}2n\,a^{+}_{n}\,z^{n}
    + \frac{2}{3}\sum_{n=2}^{\infty}b^{+}_{n}\,z^{n}
    + \frac{1}{4}\sum_{n=2}^{\infty}\biggl(\sum_{j=1}^{n-1}\bigl(3a^{+}_{j}+a^{-}_{j}\bigr)\,b^{+}_{n-j}\biggr)z^{n} &= 0, \label{eq:swe_recursion_summed_Ap}\\[3mm]
    \sum_{n=2}^{\infty}2n\,b^{-}_{n}\,z^{n}
    + \sum_{n=2}^{\infty}b^{-}_{n}\,z^{n}
    + \frac{1}{3}\sum_{n=2}^{\infty}b^{+}_{n}\,z^{n}
    + \frac{1}{4}\sum_{n=2}^{\infty}\biggl(\sum_{j=1}^{n-1}\bigl(b^{+}_{j}+3b^{-}_{j}\bigr)\,b^{-}_{n-j}\biggr)z^{n} &= 0, \label{eq:swe_recursion_summed_Bm}\\[3mm]
    \sum_{n=2}^{\infty}2n\,a^{-}_{n}\,z^{n}
    + \sum_{n=2}^{\infty}a^{-}_{n}\,z^{n}
    + \frac{1}{3}\sum_{n=2}^{\infty}a^{+}_{n}\,z^{n}
    + \frac{1}{4}\sum_{n=2}^{\infty}\biggl(\sum_{j=1}^{n-1}\bigl(a^{+}_{j}+3a^{-}_{j}\bigr)\,b^{-}_{n-j}\biggr)z^{n} &= 0. \label{eq:swe_recursion_summed_Am}
\end{align}
\end{subequations}
Each summation in~\eqref{eq:swe_recursion_summed} can be identified with a generating function or its derivative, following the same procedure as in the Burgers case~\eqref{eq:sum_ref_burgers_1}--\eqref{eq:sum_ref_burgers_2} and \eqref{eq:burgers-exterior-generating-sums}.
Simplifying gives
\begin{subequations}\label{eq:swe_ode_z}
\begin{align}
    2z\,\frac{dA_{\pm}}{dz} + \frac{1}{4}(3A_{\pm}+A_{\mp})\,B_{\pm} &= 0, \label{eq:swe_ode_A_z}\\[2mm]
    2z\,\frac{dB_{\pm}}{dz} - B_{\pm} + \frac{1}{4}(3B_{\pm}+B_{\mp})\,B_{\pm} &= -z. \label{eq:swe_ode_B_z}
\end{align}
\end{subequations}
Substituting $z = (\omega t)^{2}$, so that $2z\,\tfrac{d}{dz} = t\,\tfrac{d}{dt}$, recovers the ODE system~\eqref{eq:ABCD_ODE_swe}.

\subsection{Branches of the Riemann solution}\label{apdx_sec:Recursion_Relation_ABCD}

The ODE system~\eqref{eq:ABCD_ODE_swe} is shared by all branches of the Riemann solution. Repeating the perturbative calculation on any branch would yield the same system, with each branch distinguished by its initial condition at $t=0$. Instead, we show directly from the singular structure of the ODE system that the admissible initial data are precisely those corresponding to the branches of the Riemann solution.

Writing~\eqref{eq:ABCD_ODE_swe} as $t\,\dot{y} = f(t,y)$, with
$y = (A_+, A_-, B_+, B_-)$, the factor $f/t$ is singular at $t = 0$.
For a bounded solution $y(t) \to y_0$ as $t \to 0^+$ to exist, we need $f(0, y_0) = 0$.
If this fails, then $\dot{y} \sim f(0, y_0)/t$ near $t = 0$, which integrates to
$y \sim f(0, y_0)\log t$ and diverges. A bounded solution therefore requires
\begin{equation}\label{eq:admissibility_condition}
    f(0,y_0) = 0.
\end{equation}

The admissible initial data are the solutions of~\eqref{eq:admissibility_condition}.
Setting $z = 0$ in~\eqref{eq:swe_ode_z} gives the algebraic system
\begin{equation}\label{eq:fixed_point_system}
\begin{aligned}
    \tfrac{1}{4}(3A_+ + A_-)B_+ &= 0, \\
    \tfrac{1}{4}(3A_- + A_+)B_- &= 0, \\
    -B_+ + \tfrac{1}{4}(3B_+ + B_-)B_+ &= 0, \\
    -B_- + \tfrac{1}{4}(3B_- + B_+)B_- &= 0,
\end{aligned}
\end{equation}
whose solutions are
\begin{equation}\label{eq:fixed_points}
    (A_+, A_-, B_+, B_-) =
    \begin{cases}
        (A_+^0, A_-^0, 0, 0),                                  & \text{(i)\; spatially uniform,}    \\[2pt]
        (A_+^0, -\tfrac{1}{3}A_+^0, 0, \tfrac{4}{3}),         & \text{(ii)\; slow simple wave,}    \\[2pt]
        (-\tfrac{1}{3}A_+^0, A_+^0, \tfrac{4}{3}, 0),         & \text{(iii)\; fast simple wave,}   \\[2pt]
        (0, 0, 1, 1),                                           & \text{(iv)\; rarefaction vacuum.}
    \end{cases}
\end{equation}
Each fixed point family corresponds to a similarity solution of the unforced system, and these are the \emph{admissible} initial data for which a bounded solution of~\eqref{eq:ABCD_ODE_swe} exists at $t = 0$.
The free parameters ($A_{\pm}^0$)  in each family are pinned down by comparison with the Stoker solution~\eqref{eq:stoker_solution_RI}, yielding the initial data in~\eqref{eq:ABCD_branches}.

\begin{enumerate}\setlength{\itemsep}{6pt}
    \item[\textbf{(i)}] \textbf{Uniform depth.}
    The outer constant-depth regions of the Stoker solution~\eqref{eq:stoker_solution_RI}
    have $R_+ = 2c_0$ and $R_- = -2c_0$, giving $A_+^0 = 1$ and $A_-^0 = -1$.

    \item[\textbf{(ii)}] \textbf{Slow simple wave.}
    The fan region of the Stoker solution~\eqref{eq:stoker_solution_RI} has $R_+ = 2c_0$ constant and $R_-$ varying with $\xi = x/t$, giving $A_+^0 = 1$. Case~(iii) is its mirror and is not relevant to the slow simple-wave Stoker solution presented in this work.

    \item[\textbf{(iv)}] \textbf{Vacuum states.}
    There are two ways a vacuum can arise~\cite{liu1980vacuum}.
    \begin{itemize}\setlength{\itemsep}{4pt}
        \item \emph{Rarefaction vacuum.}
        This branch has $A_\pm = 0$ and $B_+ = B_- = 1$, and is the vacuum state used in the Stoker solution~\eqref{eq:stoker_solution_RI}. This is the only branch with no free parameter.
        Characteristics for this branch emanate from the initial point of discontinuity and diverge into the vacuum, forming a fan, so no velocity data is needed at $t = 0$ (see Figure~\ref{fig:stoker-combined}).

        \item \emph{Compressive vacuum.}
        This branch arises from the uniform family~(i) with $A_+^0 = A_-^0 = 1$,
        giving spatially constant velocity $u = 2c_0$.
        It requires prescribing an initial velocity data throughout the vacuum region.
    \end{itemize}
\end{enumerate}

The equivalent singular systems~\eqref{eq:ABCD_ODE_swe} and~\eqref{eq:swe_ode_z} belong to the class studied by de Jong and van Meurs~\cite{dejong2022}. For such systems, local analytic solutions are unique provided the eigenvalues of the linearisation about the admissible initial state satisfy a non-resonance condition, namely that none is a positive even integer. For the admissible branches~\eqref{eq:fixed_points}, 
the corresponding eigenvalues are
\[
\textrm{(i)}\ \ \{0,0,1,1\},\qquad
\textrm{(ii)}\ \ \{0,\tfrac23,-1,-1\},\qquad
\textrm{(iii)}\ \ \{0,\tfrac23,-1,-1\},\qquad
\textrm{(iv)}\ \ \left\{-1,-1,-\tfrac12,-\tfrac12\right\},
\]
respectively.
Thus the non-resonance condition is satisfied in every case. Consequently, each admissible initial state determines a unique local analytic solution.

\subsection[Invariant manifolds of the \texorpdfstring{$B_{\pm}$}{B±} subsystem]{Invariant manifolds of the $B_{\pm}$ subsystem}\label{apdx_sec:BD_system}

The closed $(B_+,B_-)$ subsystem possesses two invariant manifolds on which it reduces to scalar Riccati equations. We show that the branches of the Riemann solution lie on these manifolds for all $t\ge0$.
We show that the branches of the Riemann solution lie on these manifolds for
all $t \geq 0$.

Define
\begin{equation}
    P = B_{+}+B_{-}, \qquad M = B_{+}-B_{-}, \qquad
    Q = B_{+}B_{-} + (\omega t)^{2}.
\end{equation}
Subtracting the $B_-$ equation from the $B_+$ equation
in~\eqref{eq:ABCD_ODE_swe} gives the evolution equation for $M$, and
multiplying the equation for $B_{\pm}$ by $B_{\mp}$, then
adding, and applying the product rule gives the evolution equation for $Q$,
\begin{align}
    t\,\frac{dM}{dt} &= M\!\left(1-\tfrac{3}{4}P\right), \label{eq:PM_system_M}\\[1mm]
    t\,\frac{dQ}{dt} &= Q\left(2-P\right). \label{eq:Q_system}
\end{align}
Both equations are linear and homogeneous in their respective variables, so
$M\equiv0$ and $Q\equiv0$ are solutions. The constant-depth and vacuum
branches, and the simple-wave branches, satisfy $M(0)=0$ and $Q(0)=0$,
respectively. By the uniqueness of the local analytic solution established
above, these are the only analytic solutions through the corresponding initial
data.
Hence $M(t)\equiv0$ on the constant-depth and vacuum branches, whereas $Q(t)\equiv0$ on the simple-wave branches for all $t\ge0$.

\subsection{Proof of general translation law (Proposition~\ref{prop:translation})}\label{apdx_sec:proof_of_translation_proposition}

\begin{proof}

Let \((h,u)\) be a solution of the forced SWE~\eqref{eq:swe-parabolic}. Fix \(x_0\in\mathbb{R}\), and define
\begin{equation}\label{eq:translation_change_of_var}
    y=x-x_0\cos(\omega t), 
\end{equation}
together with
\begin{equation}\label{eq:translation_transformed_fields}
    \begin{aligned}
     H(x,t) &= h(y,t),\\
     U(x,t) &= u(y,t) - \omega x_0\sin(\omega t).
    \end{aligned}
\end{equation}

We show that \((H,U)\) also solves the forced SWE with the same bathymetry.

First, since \(y_x = 1\) and \(y_t = \omega x_0 \sin(\omega t)\), we have
\begin{align}
    H_t&=h_t+\omega x_0\sin(\omega t)h_y,\\
    (HU)_x &= (hu)_y - \omega x_0\sin(\omega t)h_y.
\end{align}
Adding the two yields
\begin{align}
    H_t+(HU)_x = h_t + (hu)_y = 0.
\end{align}
Hence the mass equation is satisfied.

\medskip

For the momentum equation,
\begin{align}
    U_t &= u_t  +  \omega x_0\sin(\omega t)u_y -\omega^2 x_0 \cos(\omega t),\\
    U_x &= u_y,\\
    H_x &= h_y.
\end{align}
Thus
\begin{align}\label{eq:translation_momentum}
    U_t + UU_x + H_x = \left(u_t + uu_y + h_y\right) -\omega^2 x_0 \cos(\omega t),
\end{align}
where, since \((h,u)\) solves the momentum equation at the point \(y\), we have \(u_t + u u_y + h_y = -\omega^2 y\).

Then expressing $y$ in terms of $x$ using~\eqref{eq:translation_change_of_var}, Eq.~\eqref{eq:translation_momentum} becomes
\begin{equation}
    U_t+U U_x+H_x=-\omega^2x.
\end{equation}
Hence the momentum equation is also satisfied.

Finally, at \(t=0\), from~\eqref{eq:translation_change_of_var}
\[
    y=x-x_0,
\]
and therefore from~\eqref{eq:translation_transformed_fields}
\begin{align}
    H(x,0)&=h(x-x_0,0),\\
    U(x,0)&=u(x-x_0,0).
\end{align}
Thus \((H,U)\) is the solution corresponding to the translated initial data. This proves Proposition~\ref{prop:translation}.
    
\end{proof}

    \section[Appendix -- Wavefront expansion and asymptotics of the run-up solution]{Wavefront expansion and asymptotics of the run-up solution}\label{apdx_sec:wavefront_higher_order}

    A wavefront expansion of the Riemann variables about the dry edge shows that the run-up solution~\eqref{eq:swe_slow_phys} is asymptotic to the vacuum edge $x = X_{\rm dry}(t;\omega)$, where it arises as the first-order truncation of the expansion. The expansion follows Camassa \emph{et al.}~\cite{camassa2022evolution}.

     \subsection{Wavefront hierarchy}\label{apdx:hierarchy}
    Introducing the wavefront variable
    \begin{equation}\label{eq:apdx_wavefront_variable}
        \chi \doteq X_{\rm dry}(t;\omega) - x
    \end{equation}
    Thus \(\chi\) measures distance behind the dry edge, with the dry edge located at \(\chi=0\). 
    
    With $\partial_x = -\partial_\chi$ and $\partial_t \mapsto \partial_t + \dot{X}_{\rm dry}\,\partial_\chi$, the diagonal form~\eqref{eq:diagonal_SWE} becomes
    \begin{equation}\label{eq:apdx_wavefront_diagonal}
        \frac{\partial R_{\pm}}{\partial t} + \bigl(\dot{X}_{\rm dry} - \lambda_{\pm}\bigr)\frac{\partial R_{\pm}}{\partial \chi} = -\omega^{2}\bigl(X_{\rm dry} - \chi\bigr).
    \end{equation}
    Expanding the Riemann variables in powers of the wavefront variable,
    \begin{equation}\label{eq:apdx_wavefront_expansion}
        R_{\pm}(\chi,t) = \sum_{n=0}^{\infty} \rho_{n}^{\pm}(t)\,\chi^{n},
    \end{equation}
    and collecting equal powers of $\chi$ gives a hierarchy of ODEs for the coefficients $\rho_n^{\pm}(t)$,
    \begin{subequations}\label{eq:apdx_hierarchy}
    \begin{align}
        \frac{d \rho_0^{\pm}}{dt} &= -\omega^{2} X_{\rm dry}, \label{eq:apdx_hierarchy_n0}\\[1mm]
        \frac{d \rho_1^{\pm}}{dt} &= \lambda_1^{\pm}\rho_1^{\pm} + \omega^{2}, \label{eq:apdx_hierarchy_n1}\\[1mm]
        \frac{d \rho_n^{\pm}}{dt} &= \sum_{k=1}^{n}(n-k+1)\,\lambda_k^{\pm}\,\rho_{n-k+1}^{\pm}, \qquad n\ge 2, \label{eq:apdx_hierarchy_nge2}
    \end{align}
    \end{subequations}
    where $\lambda_k^{\pm}(t) \doteq \tfrac{1}{4}(3\rho_k^{\pm}+\rho_k^{\mp})$.

    \subsection{Leading order and the run-up}\label{apdx:leading_runup}
    At the dry edge the depth vanishes, so $\rho_0^{+} = \rho_0^{-}$, and the edge advances at the local characteristic speed, $\dot{X}_{\rm dry} = \rho_0^{\pm}$. The leading-order equation~\eqref{eq:apdx_hierarchy_n0} together with $\dot{X}_{\rm dry} = \rho_0^{\pm}$ gives the edge equation
    \begin{equation}\label{eq:apdx_edge_ode}
        \frac{d^2{X}_{\rm dry}}{dt^2} + \omega^{2} X_{\rm dry} = 0,
    \end{equation}
    whose solution with $X_{\rm dry}(0)=0$ and edge speed $\dot{X}_{\rm dry}(0)=2c_0$ is the dry edge~\eqref{eq:swe_edge_right}, with 
    \begin{equation}
        \rho_0^{\pm}(t) = 2c_0\cos(\omega t).
    \end{equation}
    
    We now recover the first-order coefficients directly from the known run-up
    branch~\eqref{eq:swe_slow_phys}. Expressed in Riemann variables, this branch is
    \begin{subequations}\label{eq:apdx_runup_RV}
    \begin{align}
        R_{+}^{\rm run-up}(x,t)
        &=\frac{2c_{0}\cos(\tfrac{\omega t}{4})
            -\omega x\sin(\tfrac{3\omega t}{4})}
            {\cos(\tfrac{3\omega t}{4})},\\[2mm]
        R_{-}^{\rm run-up}(x,t)
        &=\frac{-2c_{0}\sin(\tfrac{\omega t}{4})
            +\omega x\cos(\tfrac{3\omega t}{4})}
            {\sin(\tfrac{3\omega t}{4})}.
    \end{align}
    \end{subequations}
    Expressing this branch in the wavefront coordinate gives
    \begin{subequations}\label{eq:apdx_runup_RV_wavefront}
    \begin{align}
        R_{+}^{\rm run-up}(\chi,t)
        &=2c_0\cos(\omega t)
            +\omega\tan\bigl(\tfrac{3\omega t}{4}\bigr)\,\chi,\\[2mm]
        R_{-}^{\rm run-up}(\chi,t)
        &=2c_0\cos(\omega t)
            -\omega\cot\bigl(\tfrac{3\omega t}{4}\bigr)\,\chi.
    \end{align}
    \end{subequations}
    Comparing~\eqref{eq:apdx_runup_RV_wavefront} with the wavefront expansion~\eqref{eq:apdx_wavefront_expansion} yields
    \begin{equation}\label{eq:apdx_rho1}
        \rho_1^{+}(t) = \omega\tan\bigl(\tfrac{3\omega t}{4}\bigr), \qquad
        \rho_1^{-}(t) = -\,\omega\cot\bigl(\tfrac{3\omega t}{4}\bigr).
    \end{equation}
    These coefficients can be verified directly to satisfy the first-order
    wavefront equations~\eqref{eq:apdx_hierarchy_n1}. Since the run-up branch has
    the wavefront form~\eqref{eq:apdx_runup_RV_wavefront} with coefficients
    satisfying the \(n=0\) and \(n=1\) wavefront hierarchy, it gives the
    first-order wavefront approximation near the dry edge. Therefore,
    \begin{equation}\label{eq:apdx_RV_asymptotic}
        R_{\pm}(\chi,t) = R_{\pm}^{\rm run-up}(\chi,t) + \mathcal{O}(\chi^{2}).
    \end{equation}
    Consequently, expressing in terms of the physical depth and velocity
    variables, it follows that
    \begin{equation}\label{eq:apdx_hu_asymptotic}
        \begin{aligned}
        h(x,t;\omega) &= \hrunup(x,t;\omega) + \mathcal{O}(\chi^{3}),\\
        u(x,t;\omega) &= \urunup(x,t;\omega) + \mathcal{O}(\chi^{2}).
        \end{aligned}
    \end{equation}
    Thus the run-up solution is asymptotic at the dry edge as \(\chi \to 0^+\).

    \subsection{Higher-order corrections}\label{apdx:higher_order}
    The coefficients $\rho_n^{\pm}$ with $n\ge 2$ satisfy the linear system 
    \begin{equation}\label{eq:apdx_higher_linear}
        \frac{\rm d}{{\rm d}t}\begin{pmatrix}\rho_n^{+}\\[1mm]\rho_n^{-}\end{pmatrix}
        =
        \frac{3\omega}{4}\begin{pmatrix}
            (n+1)\tan \tfrac{3\omega t}{4} - \tfrac{n}{3}\cot \tfrac{3\omega t}{4} & \tfrac{1}{3}\tan \tfrac{3\omega t}{4} \\[2mm]
            -\tfrac{1}{3}\cot \tfrac{3\omega t}{4} & -(n+1)\cot \tfrac{3\omega t}{4} + \tfrac{n}{3}\tan \tfrac{3\omega t}{4}
        \end{pmatrix}
        \begin{pmatrix}\rho_n^{+}\\[1mm]\rho_n^{-}\end{pmatrix}
        + \boldsymbol{F}_n,
    \end{equation}
    where the forcing vector $\boldsymbol{F}_n$ collects the lower-order contributions
    from~\eqref{eq:apdx_hierarchy_nge2},
    \begin{equation}\label{eq:apdx_higher_forcing}
        \boldsymbol{F}_n \doteq \sum_{k=2}^{n-1}(n-k+1)
        \begin{pmatrix}
            \lambda_k^{+}\,\rho_{n-k+1}^{+}\\[2mm]
            \lambda_k^{-}\,\rho_{n-k+1}^{-}
        \end{pmatrix}, \qquad \boldsymbol{F}_2 \doteq \boldsymbol{0},
    \end{equation}
    with $\lambda_k^{\pm} \doteq (3\rho_k^{\pm} + \rho_k^{\mp})/2$. The forcing 
    is not present at $n=2$, so the first correction beyond the run-up solution is governed by the homogeneous system alone.

    \paragraph{\textbf{Normalised equations}}
    The $\omega$-scaling of $\rho_n^{\pm}$ is fixed by the recursive structure of the 
    hierarchy. The first-order coefficients~\eqref{eq:apdx_rho1} each carries one factor of $\omega$, so $\rho_1^{\pm}(t) = \omega\,\tilde{\rho}_1^{\pm}(\tau)$ with 
    $\tau \doteq \omega t$ and $\tilde{\rho}_1^{\pm}$ independent of $\omega$. This 
    factorisation propagates to every order,
    \begin{equation}\label{eq:apdx_ansatz}
        \rho_n^{\pm}(t) \;=\; \omega^n\,\tilde{\rho}_n^{\pm}(\tau), \qquad \tau = \omega t,
    \end{equation}
    with $\tilde{\rho}_n^{\pm}$ depending on $\tau$ alone. To verify this, assume \eqref{eq:apdx_ansatz} for all $1 \le k \le n-1$. Then 
    $\lambda_k^{\pm} = \omega^k\,\tilde{\lambda}_k^{\pm}(\tau)$ with 
    $\tilde{\lambda}_k^{\pm} \doteq (3\tilde{\rho}_k^{\pm} + \tilde{\rho}_k^{\mp})/2$, 
    and every term of the forcing in~\eqref{eq:apdx_higher_forcing} produces the 
    same total power $\omega^k \cdot \omega^{n-k+1} = \omega^{n+1}$. The forcing 
    therefore factorises as
    \begin{equation}\label{eq:apdx_forcing_scaling}
        \boldsymbol{F}_n \;=\; \omega^{n+1}\,\tilde{\boldsymbol{F}}_n(\tau), \qquad
        \tilde{\boldsymbol{F}}_n(\tau) \;\doteq\; \sum_{k=2}^{n-1}(n-k+1)
        \begin{pmatrix}\tilde{\lambda}_k^{+}\,\tilde{\rho}_{n-k+1}^{+}\\[1mm]
        \tilde{\lambda}_k^{-}\,\tilde{\rho}_{n-k+1}^{-}\end{pmatrix},
    \end{equation}
    with $\tilde{\boldsymbol{F}}_n$ independent of $\omega$. Substituting 
    \eqref{eq:apdx_ansatz} and \eqref{eq:apdx_forcing_scaling} into 
    \eqref{eq:apdx_higher_linear} and using $d/dt = \omega\,d/d\tau$, every term carries 
    a common factor of $\omega^{n+1}$ that cancels, giving a system that is independent of $\omega$,
    \begin{equation}\label{eq:apdx_tau_system}
        \frac{\rm d}{{\rm d}\tau}\begin{pmatrix}\tilde{\rho}_n^{+}\\[1mm]\tilde{\rho}_n^{-}\end{pmatrix}
        =
        \frac{3}{4}\begin{pmatrix}
            (n+1)\tan \tfrac{3\tau}{4} - \tfrac{n}{3}\cot \tfrac{3\tau}{4} & \tfrac{1}{3}\tan \tfrac{3\tau}{4} \\[2mm]
            -\tfrac{1}{3}\cot \tfrac{3\tau}{4} & -(n+1)\cot \tfrac{3\tau}{4} + \tfrac{n}{3}\tan \tfrac{3\tau}{4}
        \end{pmatrix}
        \begin{pmatrix}\tilde{\rho}_n^{+}\\[1mm]\tilde{\rho}_n^{-}\end{pmatrix}
        + \tilde{\boldsymbol{F}}_n(\tau).
    \end{equation}
    This confirms~\eqref{eq:apdx_ansatz} at order $n$.

    \paragraph{\textbf{Leading behaviour at each order}}
We now consider the behaviour of the system near $\tau = 0$. Using 
    $\tan\tfrac{3\tau}{4} = \tfrac{3\tau}{4} + O(\tau^3)$ and 
    $\cot\tfrac{3\tau}{4} = \tfrac{4}{3\tau} + O(\tau)$, the homogeneous part of 
    \eqref{eq:apdx_tau_system} reduces to
    \begin{equation}\label{eq:apdx_leading}
        \frac{\rm d}{{\rm d}\tau}\begin{pmatrix}\tilde{\rho}_n^{+}\\[1mm]\tilde{\rho}_n^{-}\end{pmatrix}
        =
        \frac{1}{\tau}
        \begin{pmatrix}
            -\tfrac{n}{3} & 0 \\[2mm]
            -\tfrac{1}{3} & -(n+1)
        \end{pmatrix}
        \begin{pmatrix}\tilde{\rho}_n^{+}\\[1mm]\tilde{\rho}_n^{-}\end{pmatrix}
        + O(\tau),
    \end{equation}
    which has eigenvalues $-\tfrac{n}{3}$ and $-(n+1)$ and eigenvectors 
    $\boldsymbol{e}_1 = \bigl(1, -\tfrac{1}{2n+3}\bigr)^{\top}$ and 
    $\boldsymbol{e}_2 = (0, 1)^{\top}$. We use \eqref{eq:apdx_ansatz} to write the 
    general homogeneous solution in terms of $\rho_n^{\pm}$,
    \begin{equation}\label{eq:apdx_euler_solution}
        \begin{pmatrix}\rho_n^{+}\\[1mm]\rho_n^{-}\end{pmatrix}
        \;\sim\; C_n\,\omega^n\,\tau^{-n/3}\begin{pmatrix}1\\[1mm]-\tfrac{1}{2n+3}\end{pmatrix}
        + D_n\,\omega^n\,\tau^{-(n+1)}\begin{pmatrix}0\\[1mm]1\end{pmatrix},
    \end{equation}
    where $C_n$ and $D_n$ are order-dependent constants of integration independent of $\omega$. The 
    branch $\tau^{-(n+1)}$ gives $\rho_n^{\pm} \propto t^{-n-1}/\omega$, which diverges as $\omega \to 0$ at fixed $t$. 
    Requiring the corrections to vanish in this flat bathymetry limit sets 
    $D_n = 0$. Defining
    \[
        C_n^+ \doteq C_n, \qquad
        C_n^- \doteq -\frac{C_n}{2n+3}.
    \]
    The admissible branch may then be written as
    \begin{equation}\label{eq:apdx_fractional_power}
        \rho_n^{\pm}(t) \;\sim\; \frac{C_n^{\pm}}{t^n}\,(\omega t)^{2n/3}
        \qquad (\omega \to 0).
    \end{equation}
    This shows that the higher-order wavefront coefficients contain terms with
    fractional powers of \(\omega t\). In contrast, the regular perturbation
    series in \(\omega^2\) yielded solutions with power series in integer powers
    of \(\omega t\). Consequently, the regular perturbative framework, which is
    the principal contribution of this work, captures the run-up branch but does
    not capture the higher-order wavefront dynamics beyond it.

    Computing further terms by a Frobenius expansion gives the more detailed local
    structure
    \begin{equation}\label{eq:apdx_frobenius_structure}
        \rho_n^\pm(t)
        =
        \frac{C_n^\pm}{t^n}(\omega t)^{2n/3}
        \left[
        1+a_{n,1}^\pm(\omega t)^2
          +a_{n,2}^\pm(\omega t)^4+\cdots
        \right],
        \qquad n\geq 2,
    \end{equation}
    where the coefficients \(a_{n,j}^\pm\) are determined recursively by the
    Frobenius expansion. The overall amplitude \(C_n\), however, is not fixed
    by this local analysis. It represents a remaining degree of freedom that must
    be selected by matching the wavefront expansion to the solution away from the
    dry edge, in particular at the left state. 
    A detailed analysis of the higher-order solution structure and the associated matching problem remains an open direction for future work.

\end{appendix}

%% file: backmatter.tex
\begin{Backmatter}

\paragraph{Acknowledgments}
Two of the co-authors (MAH and BI) partook in Nonlinear Waves and Integrable Systems meeting that was held in Boulder, Colorado in June 2025 on the occasion of the 80th birthday of Mark J. Ablowitz. This event reflected the profound, broad and sustained contributions that Prof. Ablowitz has made in the fields of nonlinear waves, integrable systems, and asymptotics. We are among those people who have benefited immensely from Mark, scientifically and personally, in ways that are too many to describe.

The present work also benefited from discussions held during the Isaac Newton Institute Satellite Programme ``Emergent phenomena in nonlinear dispersive waves''. BI, MAH, PS, and SS thank the Isaac Newton Institute for Mathematical Sciences, Cambridge, for support and hospitality during this programme, which was held at Northumbria University and Newcastle University in Newcastle, UK, from 16 July to 22 August 2024. This programme was supported by EPSRC grant EP/R014604/1.

\paragraph{Funding Statement}
This research was supported by a grant from the National Science Foundation doi.org (DMS-2306319). 

\paragraph{Competing Interests}
None.

\printbibliography

\end{Backmatter}

%% file: references.bib
@article{ritter1892fortpflanzung,
  title={Die Fortpflanzung der Wasserwellen},
  author={Ritter, August},
  journal={Zeitschrift des Vereines Deutscher Ingenieure},
  volume={36},
  number={2},
  pages={947--954},
  year={1892}
}

@book{stoker2019water,
  title={Water Waves: The Mathematical Theory with Applications},
  author={Stoker, James Johnston},
  year={2019},
  publisher={Courier Dover Publications}
}

@article{stoker1948formation,
  title={The formation of breakers and bores the theory of nonlinear wave propagation in shallow water and open channels},
  author={Stoker, James J},
  journal={Communications on Pure and Applied Mathematics},
  volume={1},
  number={1},
  pages={1--87},
  year={1948},
  publisher={Wiley Online Library}
}

@article{carrier1958water,
  title={Water waves of finite amplitude on a sloping beach},
  author={Carrier, George F and Greenspan, Harvey P},
  journal={Journal of Fluid Mechanics},
  volume={4},
  number={1},
  pages={97--109},
  year={1958},
  publisher={Cambridge University Press}
}

@article{ovsyannikov1979two,
  title={Two-layer ``shallow water'' model},
  author={Ovsyannikov, Lev Vasil'evich},
  journal={Journal of Applied Mechanics and Technical Physics},
  volume={20},
  number={2},
  pages={127--135},
  year={1979},
  publisher={Springer}
}

@article{thacker1981some,
  title={Some exact solutions to the nonlinear shallow-water wave equations},
  author={Thacker, William Carlisle},
  journal={Journal of Fluid Mechanics},
  volume={107},
  pages={499--508},
  year={1981},
  publisher={Cambridge University Press}
}

@article{pelinovsky1992exact,
  title={Exact analytical solutions of nonlinear problems of tsunami wave run-up on slopes with different profiles},
  author={Pelinovsky, EN and Mazova, R Kh},
  journal={Natural Hazards},
  volume={6},
  number={3},
  pages={227--249},
  year={1992},
  publisher={Springer}
}

@article{kanouglu2006initial,
  title={Initial value problem solution of nonlinear shallow water-wave equations},
  author={K{\^a}no{\u{g}}lu, Utku and Synolakis, Costas},
  journal={Physical Review Letters},
  volume={97},
  number={14},
  pages={148501},
  year={2006},
  publisher={APS}
}

@article{didenkulova2011nonlinear,
  title={Nonlinear wave evolution and runup in an inclined channel of a parabolic cross-section},
  author={Didenkulova, Ira and Pelinovsky, Efim},
  journal={Physics of Fluids},
  volume={23},
  number={8},
  year={2011},
  publisher={AIP Publishing}
}

@article{ezersky2013physical,
  title={Physical simulation of resonant wave run-up on a beach},
  author={Ezersky, Alexander and Abcha, Nizar and Pelinovsky, Efim},
  journal={Nonlinear Processes in Geophysics},
  volume={20},
  number={1},
  pages={35--40},
  year={2013},
  publisher={Copernicus Publications G{\"o}ttingen, Germany}
}

@article{aksenov2016conservation,
title = {Conservation laws and symmetries of the shallow water system above rough bottom},
author = {Aksenov, A V and Druzhkov, K P},
journal = {Journal of Physics: Conference Series},
year = {2016},
publisher = {IOP Publishing},
volume = {722},
number = {1},
pages = {012001}
}

@article{aksenov2020conservation,
  title={Conservation laws of the equation of one-dimensional shallow water over uneven bottom in Lagrange's variables},
  author={Aksenov, Alexander V and Druzhkov, Konstantin P},
  journal={International Journal of Non-Linear Mechanics},
  volume={119},
  pages={103348},
  year={2020},
  publisher={Elsevier}
}

@article{el2016dam,
  title={Dam break problem for the focusing nonlinear Schr{\"o}dinger equation and the generation of rogue waves},
  author={El, Gennady A and Khamis, Eduardo G and Tovbis, Alexander},
  journal={Nonlinearity},
  volume={29},
  number={9},
  pages={2798--2836},
  year={2016},
  publisher={IOP Publishing}
}

@article{xu2017dispersive,
  title={Dispersive dam-break flow of a photon fluid},
  author={Xu, Gang and Conforti, Matteo and Kudlinski, Alexandre and Mussot, Arnaud and Trillo, Stefano},
  journal={Physical review letters},
  volume={118},
  number={25},
  pages={254101},
  year={2017},
  publisher={APS}
}

@article{camassa2019singularity,
  title={Singularity formation as a wetting mechanism in a dispersionless water wave model},
  author={Camassa, Roberto and Falqui, Gregorio and Ortenzi, Giovanni and Pedroni, Marco and Pitton, Giuseppe},
  journal={Nonlinearity},
  volume={32},
  number={10},
  pages={4079--4116},
  year={2019},
  publisher={IOP Publishing}
}

@article{camassa2020vacuum,
  title={On the ``Vacuum'' Dam-Break Problem: exact solutions and their long time asymptotics},
  author={Camassa, Roberto and Falqui, Gregorio and Ortenzi, Giovanni and Pedroni, Marco and Pitton, Giuseppe},
  journal={SIAM Journal on Applied Mathematics},
  volume={80},
  number={1},
  pages={44--70},
  year={2020},
  publisher={SIAM}
}

@article{camassa2022evolution,
  title={Evolution of interface singularities in shallow water equations with variable bottom topography},
  author={Camassa, R and D'Onofrio, R and Falqui, G and Ortenzi, Giovanni and Pedroni, Marco},
  journal={Studies in Applied Mathematics},
  volume={148},
  number={4},
  pages={1439--1476},
  year={2022},
  publisher={Wiley Online Library}
}

@article{dieli2024observation,
  title={Observation of two-dimensional dam break flow and a gaseous phase of solitons in a photon fluid},
  author={Dieli, Ludovica and Pierangeli, Davide and DelRe, Eugenio and Conti, Claudio},
  journal={Physical Review Letters},
  volume={133},
  number={18},
  pages={183801},
  year={2024},
  publisher={APS}
}

@article{sharan2025breaking,
  title={Breaking a superfluid harmonic dam: Observation and theory of Riemann invariants and accelerating sonic horizons},
  author={Sharan, Shashwat and Sorribes, Judith Gonzalez and Sprenger, Patrick and Hoefer, Mark A and Engels, P and Ilan, Boaz and Mossman, ME},
  journal={Physical Review Letters},
  volume={135},
  number={19},
  pages={193402},
  year={2025},
  publisher={APS}
}

@article{jenssen2017exact,
  title={On exact solutions of rarefaction-rarefaction interactions in compressible isentropic flow},
  author={Jenssen, Helge Kristian},
  journal={Journal of Mathematical Fluid Mechanics},
  volume={19},
  number={4},
  pages={685--708},
  year={2017},
  publisher={Springer}
}

@article{pang2012collision,
  title={Collision of Two Centered Rarefaction Waves for Isentropic Magnetogasdynamics.},
  author={Pang, Yicheng and Gu, Hongbo and Yang, Hanchun},
  journal={Southeast Asian Bulletin of Mathematics},
  volume={36},
  number={3},
  year={2012}
}

@article{liu1980vacuum,
  title={On the vacuum state for the isentropic gas dynamics equations},
  author={Liu, T-P and Smoller, Joel A},
  journal={Advances in Applied Mathematics},
  volume={1},
  number={4},
  pages={345--359},
  year={1980},
  publisher={Academic Press}
}

@article{dorodnitsyn2021discrete,
  title={Discrete shallow water equations preserving symmetries and conservation laws},
  author={Dorodnitsyn, VA and Kaptsov, EI},
  journal={Journal of Mathematical Physics},
  volume={62},
  number={8},
  year={2021},
  publisher={AIP Publishing}
}

@article{chirkunov2014exact,
  title={Exact solutions of one-dimensional nonlinear shallow water equations over even and sloping bottoms},
  author={Chirkunov, Yu A and Dobrokhotov, S Yu and Medvedev, Sergey Borisovich and Minenkov, Dmitrii Sergeevich},
  journal={Theoretical and Mathematical Physics},
  volume={178},
  number={3},
  pages={278--298},
  year={2014},
  publisher={Springer}
}

@article{dejong2022,
  title={Uniqueness of local, analytic solutions to singular ODEs},
  author={de Jong, Thomas Geert and van Meurs, Patrick},
  journal={Acta Applicandae Mathematicae},
  volume={180},
  number={1},
  pages={14},
  year={2022},
  publisher={Springer}
}

@article{delestre2013swashes,
  title={SWASHES: a compilation of shallow water analytic solutions for hydraulic and environmental studies},
  author={Delestre, Olivier and Lucas, Carine and Ksinant, Pierre-Antoine and Darboux, Fr{\'e}d{\'e}ric and Laguerre, Christian and Vo, T-N-Tuoi and James, Francois and Cordier, St{\'e}phane},
  journal={International Journal for Numerical Methods in Fluids},
  volume={72},
  number={3},
  pages={269--300},
  year={2013},
  publisher={Wiley Online Library}
}
